\documentclass{article}

\usepackage[main, final]{neurips_2026}

\usepackage[utf8]{inputenc}
\usepackage[T1]{fontenc}
\usepackage{hyperref}
\usepackage{url}
\usepackage{booktabs}
\usepackage{amsfonts}
\usepackage{nicefrac}
\usepackage{microtype}
\usepackage{xcolor}

\usepackage{graphicx}
\usepackage{subcaption}
\usepackage{amsmath}
\usepackage{amssymb}
\usepackage{mathtools}
\usepackage{amsthm}
\usepackage{bbm}
\usepackage[capitalize,noabbrev]{cleveref}

\DeclareMathOperator*{\argmin}{arg\,min}

\theoremstyle{plain}
\newtheorem{theorem}{Theorem}[section]

\newtheorem{corollary}[theorem]{Corollary}

\theoremstyle{definition}
\newtheorem{definition}[theorem]{Definition}
\newtheorem{assumption}[theorem]{Assumption}
\theoremstyle{remark}
\newtheorem{remark}[theorem]{Remark}

\title{Beyond IPS: Reliable Counterfactual Evaluation in Multi-Stage Ad Systems without Logged Propensities}

\author{%
  Mohsen Malmir\thanks{Corresponding author} \\
  Meta Platforms, Inc. \\
  \texttt{mohsenm@meta.com} \\
  \And
  Mohamed A.\ Radwan \\
  Meta Platforms, Inc. \\
  \texttt{attia@meta.com} \\
  \And
  Houssam Nassif \\
  Meta Platforms, Inc. \\
  \texttt{houssamn@meta.com} \\
  \And
  Murat Ali Bayir \\
  Meta Platforms, Inc. \\
  \texttt{mbayir@meta.com} \\
}

\begin{document}

\maketitle

\begin{abstract}

\par{
Offline evaluation of ad-ranking policies in multi-stage delivery systems is challenging because action propensities are often unavailable, limiting inverse propensity scoring (IPS). Reward-model-based estimators such as the Direct Method (DM) avoid propensities but can be biased when candidate policies induce impression distributions that differ from the logged data. We formalize DM-based counterfactual replay and propose mixture-weighted reward modeling, which reweights logged data toward a mixture of candidate-policy distributions using density ratios from a multiclass domain classifier. We bound each candidate's value-estimation error by overlap, cold-start, and replay terms, and derive an asymptotic rule that uses the cold-start mass to decide whether to trust offline selection or route candidates to online A/B testing. On an ads-delivery simulator and on production traffic from a large-scale CTR ranking system, the mixture-weighted reward model improves policy selection over its unweighted counterpart in both settings, with the largest gains at low overlap in simulation. Combined with a doubly robust correction, it attains the highest Top-1 accuracy in production and at low and medium overlap in simulation.
}

\end{abstract}

\section{Introduction}

\par{
Offline evaluation of ranking models is a critical safeguard in large-scale advertising and recommendation systems, where deploying a poorly performing model can degrade user experience and business outcomes. The industry invests heavily in developing complex deep neural network (DNN) based ranking models, but before any new policy is launched, it must be rigorously assessed using historical production logs to ensure reliability and safety~\cite{yi2013predictiveofflineonline}. More broadly, offline evaluation plays two complementary roles in modern ad systems. First, it acts as a risk mitigation mechanism by identifying problematic models before they are exposed to real user traffic. Second, it serves as a scaling factor for online experimentation: large-scale ad systems typically have limited A/B testing capacity, making it infeasible to evaluate many candidate model changes online at once. Offline evaluation therefore provides a high-throughput filtering stage, allowing practitioners to prioritize the most promising policies for scarce online experimentation resources~\citep{radwan2024counterfactual, Malmir2026ImpressionShare}.
}

\par{
Accurate offline evaluation in modern advertising systems is fundamentally challenging due to the complexity of the underlying data generating process. In modern ad platforms, exposure is determined by a multi-stage pipeline that combines multiple ranking models and platform policies including auctions, filtering, and business rules rather than a single stochastic ranking policy~\citep{BottouPQCCPRSS13,facebook_practical_ctr}. Consequently, the observed data reflects a time-varying mixture of production policies~\citep{bayir2019genie}. Action propensities are typically neither logged nor well-defined, rendering inverse propensity based estimators inapplicable in practice~\citep{yi2013predictiveofflineonline,breakdeterminismauction,recsyscausalsurvey2022}.
}

\par{
Predictive metrics such as log-loss, AUROC, and Relative Information Gain (RIG) measure accuracy on logged data but do not estimate the value of the actions a new policy would take~\citep{facebook_practical_ctr,wide_deep,yi2013predictiveofflineonline}. Counterfactual evaluation estimates this value on historical contexts~\citep{zeng2021causal}, either by reweighting logged data with inverse propensity scoring (IPS), its self-normalized variants~\citep{BottouPQCCPRSS13}, and doubly robust (DR) estimators~\citep{dudik2011doublyrobust}, including extensions to deterministic logging and low overlap~\citep{breakdeterminismauction,khan2024beyondoverlap}, or by simulation and direct replay with learned response models~\citep{bayir2019genie,graepel2010adpredictor,jeunen2022auctiongym}. IPS-style estimators require logged propensities, which unlogged \emph{post-ranking system operations} in multi-stage ad systems (e.g., caching, audience targeting) make unavailable, while simulation and replay must predict outcomes for actions that were never logged, often where overlap with the logs is limited~\citep{canamares2020offlineoptions,khan2024beyondoverlap}.
}
\par{
To address these limitations, recent work has focused on reward modeling approaches that learn a predictive model of user response from logged data, and use it to estimate the value of new policies by replaying them on historical contexts~\citep{zeng2021causal,jeunen2021pessimistic,swaminathan2015counterfactual}.
When the distribution of counterfactual actions under a target policy differs substantially from the distribution induced by the logging mixture, domain shift can make reward predictions unreliable on the actions that matter for evaluation. In this case, a reward model trained on logged data may exhibit systematic misprediction on target-policy actions and produce biased policy value estimates~\citep{khan2024beyondoverlap}.

\par{
In this paper, we develop a reward-modeling framework for offline counterfactual evaluation under policy-induced distribution shift, and analyze its reliability theoretically and empirically. Our contributions are:}

(1)~\textbf{Mixture-weighted reward model training without propensities.} We introduce a practical method for direct-method counterfactual evaluation under multi-stage delivery: a single reward model trained to be accurate across a set of candidate policies by reweighting the supervised loss with mixture-induced density ratios estimated from a multiclass domain classifier.\\
(2)~\textbf{Reliability of replay Direct Method.} We derive a value error bound that holds with high probability for the Direct Method estimator used in our experiments. The bound separates reward model error where the logging data provide support, error from impressions with limited overlap, and error due to a finite evaluation sample. It makes explicit the finite-sample tradeoff between reward model error on well-supported impressions and error from impressions with limited logging support. Under consistency, only the error from impressions with little or no logging support remains asymptotically, motivating a rule for deciding when offline evaluation is sufficient and when online A/B testing is needed.\\
(3)~\textbf{Empirical evaluation on synthetic and production data.} We benchmark six estimators on a configurable ads simulator and on production traffic from a large-scale CTR ranking system, and show that the mixture-weighted reward model improves policy selection over its unweighted counterpart in both settings, with the largest gains at low overlap.
}
\section{Literature Review}
\par{
A substantial body of applied research evaluates ad ranking and click-through prediction models using predictive metrics computed over logged impressions, such as log-loss, AUROC, AUC-lift, and NDCG. These approaches have powered large-scale ad systems. Examples include Facebook’s logistic-regression and feature-engineering pipeline \citep{facebook_practical_ctr}, the adoption of deep architectures such as Wide \& Deep for CTR estimation \citep{wide_deep}, and multi-task models such as ESMM that jointly learn click and conversion probabilities on display-ad traffic \citep{esmm}.
While effective at measuring predictive performance under the production exposure distribution, such offline metrics do not capture a model’s counterfactual effect on ad delivery, user experience, or revenue.
This motivates the need for counterfactual and causal evaluation methods that estimate what would happen if a new ranking model were deployed.
}

\par{
A natural alternative to predictive metrics is off-policy evaluation (OPE), which estimates the value of a new ranking policy using data logged under a different policy.
Classical OPE methods include IPS and its self-normalized variants \citep{BottouPQCCPRSS13}, doubly robust estimators \citep{dudik2011doublyrobust}, and model-based or direct-method approaches \citep{swaminathan2015counterfactual}, as well as more recent advances addressing deterministic logging and low-overlap regimes \citep{breakdeterminismauction, khan2024beyondoverlap}. Propensity-based methods provide unbiased or partially identified estimates when key assumptions hold: the logged propensities must be known, the target policy must have support under the logging policy, and no unobserved transformations may occur between action selection and outcome. In multi-stage ad systems, these conditions are violated because downstream \emph{post-ranking system operations}
(e.g., auctions, audience targeting and business rules) transform the ranking model’s output in ways that are unlogged, nonlinear, and not recoverable from impression logs. Consequently, OPE estimators cannot identify the true value of a ranking model, limiting their reliability for offline evaluation.

}

\par{
A complementary line of work seeks to evaluate or optimize ad systems via simulation, typically by replaying historical auctions and combining logged data with predictive models. Early systems at Bing combined web-scale Bayesian CTR models with auction replay to study policy changes \citep{graepel2010adpredictor}. Later work formalized this idea as auction-based offline evaluation: \citet{yi2013predictiveofflineonline} proposed simulated metrics by rerunning sponsored-search auctions offline and estimating user responses from historic clicks; \citet{BottouPQCCPRSS13} discussed counterfactual ad-placement evaluation via a predictor that runs the auction and a click model trained offline; \citet{bayir2019genie} introduced an “open box” counterfactual policy estimator for optimizing sponsored-search marketplaces; and \citet{bai2018practicalauction} built an auction simulation system to optimize constrained mechanism parameters on an e-commerce platform. \citet{li2023practical} frames large-scale sponsored-search auctions as explicitly multi-objective, optimizing tradeoffs among platform, advertiser, and user objectives via an online optimization formulation. More recently, open-source and industrial simulators such as AuctionGym \citep{jeunen2022auctiongym} and Uber’s ads simulation framework \citep{uber2023adssim} provide environments for testing bidding and allocation strategies under synthetic or replayed traffic. These approaches are powerful when the auction logic and downstream system operations are fully specified and stable, so that the simulator achieves high fidelity to production behavior.
}
\section{Approach}
\label{sec:approach}

\subsection{Problem Setup}
\label{sec:problem_setup}

\par{
Let $x\in\mathcal{X}$ denote a context (user features) and let $a\in\mathcal{A}(x)$ denote an action (ad selection) available in context~$x$.
A \emph{(ranking) policy} $\pi$ is a deterministic mapping from contexts to actions: $\pi(x)$ denotes the action selected by $\pi$ on context $x$.
We call the pair $z=(x,a)$ an \emph{impression}, write $\mathcal{Z} = \{(x,a) : x \in \mathcal{X},\, a \in \mathcal{A}(x)\}$ for the \emph{impression space}, and let $r^\star(z)$ denote the unknown expected reward (e.g., click probability). We assume rewards are bounded: $r^\star(z) \in [0, R_{\max}]$ for all $z \in \mathcal{Z}$ and we take $R_{\max} = 1$ throughout, as our rewards are click probabilities in CTR prediction. Each policy $\pi_j$ induces a distribution $P_{\pi_j}$ on $\mathcal{Z}$; we write $\mathrm{supp}(P_{\pi_j})$ for the set of impressions that $\pi_j$ produces with positive density. We observe $n$ logged impressions
\begin{equation}
\mathcal{D}=\{(z_i,r_i)\}_{i=1}^n, \qquad z_i=(x_i,a_i),
\label{eq:logged_data}
\end{equation}
collected under a production \emph{logging policy}~$\pi_0$, where $r_i$ is the realized reward for impression~$z_i$.
Let $P_X$ denote the marginal distribution over contexts in~$\mathcal{D}$.
}

\par{
We wish to evaluate $N$ candidate policies $\pi_1,\dots,\pi_N$ using only the logged data~$\mathcal{D}$.
Each candidate policy~$\pi_j$ may select a different action than $\pi_0$ on the same context, inducing a different distribution over impressions $P_{\pi_j}$.
The \emph{value} of policy~$\pi_j$ is the expected reward under its induced impressions,
\begin{equation}
V_j \triangleq \mathbb{E}_{x\sim P_X}\!\left[r^\star\!\big(x,\pi_j(x)\big)\right].
\label{eq:policy_value}
\end{equation}
In particular, $V_0$ is the value of the current production (logging) policy.
Our analysis assumes that each candidate policy can be replayed on the logged contexts and that the reward distribution for a given impression is unchanged between the logging and candidate-policy distributions.
In practice, only a small number of candidate policies can be validated through online A/B testing due to limited experimentation capacity, making it critical to select the best policy among $N$ candidates using offline estimates.
Our goal is therefore not only to estimate $V_j$ accurately, but to do so well enough to reliably rank the candidates and minimize the regret of selecting a suboptimal policy for deployment in these limited A/B testing slots.
}

\subsection{Direct Method based Approach}
\label{sec:direct_method}
\par{
Deterministic logging, missing propensities, and unobserved multi-stage system transformations make propensity-based estimators (e.g., IPS and its variants) difficult to apply reliably. We therefore resort to the \emph{Direct Method (DM)}~\citep{dudik2011doublyrobust,swaminathan2015counterfactual,bayir2019genie}: for each candidate~$\pi_j$, we learn a reward model~$\hat r$ from logged data, apply $\pi_j$ to each logged context~$x$, and estimate the value as
\begin{equation}
\hat V_j
\;\triangleq\;
\mathbb{E}_{x\sim P_X}\!\left[\hat r\!\big(x,\,\pi_j(x)\big)\right].
\label{eq:dm_estimator}
\end{equation}
In practice, over a held-out evaluation split $\mathcal{D}_{\mathrm{eval}}\subset\mathcal{D}$ of size $m=|\mathcal{D}_{\mathrm{eval}}|$, we use the empirical replay mean $\widehat V_j^{\mathrm{DM}}=m^{-1}\sum_{(z_i,r_i)\in\mathcal{D}_{\mathrm{eval}}}\widehat r(x_i,\pi_j(x_i))$, which is the estimator used in our experiments.
The key limitation of DM is \emph{policy-induced distribution shift}: $\hat r$ is trained on impressions from~$P_{\pi_0}$ but evaluated on impressions from~$P_{\pi_j}$, which can differ substantially~\citep{shimodaira2000covariate,sugiyama2012density}. When $\hat r$ is inaccurate where $P_{\pi_j}$ places mass but $P_{\pi_0}$ does not, $\hat V_j$ is biased; if this bias differs across candidates, it can distort their ranking and lead to selecting a suboptimal policy.
}

\subsection{Mixture-Weighted Reward Model Training}
\label{sec:mixture_weighted}

\par{
To correct the policy-induced distribution shift between the logging data and the candidate policies, we replay all $N$ candidates on the logged contexts and train a \emph{single shared} reward model whose loss is weighted to target a \emph{mixture} of the candidate-policy distributions, rather than the logging distribution alone. The target mixture distribution over impressions is
\begin{equation}
P_{\mathrm{mix}}
\;\triangleq\;
\sum_{j=1}^{N}
\beta_j\,P_{\pi_j},
\qquad
\beta_j>0,\quad \sum_{j=1}^{N}\beta_j=1,
\label{eq:pmix_def}
\end{equation}
where the coefficients $\beta_1,\dots,\beta_N$ encode how strongly each candidate distribution is emphasized during training. Under no prior information about which candidates are more important, the natural default is uniform weighting: $\beta_j = 1/N$, which we use throughout. Let $p_0(z)$ and $p_j(z)$ denote the density functions of $P_{\pi_0}$ and $P_{\pi_j}$, respectively, and define the mixture density $p_{\mathrm{mix}}(z) \triangleq \sum_{j=1}^{N} \beta_j\,p_j(z)$.
}

\par{
On the region where $p_0(z)>0$, define the per-impression \emph{mixture-to-logging density ratio} and the \emph{per-policy density ratios}:
\begin{equation}
\Omega_j(z) \;\triangleq\; \frac{p_j(z)}{p_0(z)},
\qquad
\Omega(z) \;\triangleq\; \frac{p_{\mathrm{mix}}(z)}{p_0(z)} \;=\; \sum_{j=1}^{N} \beta_j\,\Omega_j(z).
\label{eq:mixture_weight}
\end{equation}
We call $\Omega_j(z)$ the \emph{true per-policy density ratio} of candidate $\pi_j$ to the logging policy, and $\Omega(z)$ the mixture ratio, which serves as a per-impression importance weight: it up-weights logged impressions that are likely under the candidate policies and down-weights those that are only relevant to the logging policy.
}

\par{
Let $\mathcal F$ denote the reward-model class, let $f\in\mathcal F$ be a candidate reward predictor, and let $\ell$ be a per-sample loss with realized reward $Y$. Using the density ratio~$\Omega(z)$, define the \emph{mixture-weighted} risk
\begin{equation}
\mathcal{L}_{\Omega}(f)
\;\triangleq\;
\mathbb{E}_{Z\sim P_{\pi_0}}\!\left[\Omega(Z)\,\ell\!\big(f(Z),\,Y\big)\right]
\;=\;
\mathbb{E}_{Z\sim P_{\mathrm{mix}}}\!\left[\ell\!\big(f(Z),\,Y\big)\,\mathbf 1\{p_0(Z)>0\}\right].
\label{eq:weighted_loss}
\end{equation}
The learned reward model is denoted by $\widehat r\in\mathcal F$. Equation~\eqref{eq:weighted_loss} targets the candidate-mixture risk restricted to the support of the logging distribution.
In practice, $\Omega(z)$ is unknown and must be estimated from data, as described in Section~\ref{sec:domain_classifier}.
}

\subsection{Estimating Density Ratios via Domain Classification}
\label{sec:domain_classifier}

\par{
The per-policy density ratios $p_j(z)/p_0(z)$ in~\eqref{eq:mixture_weight} cannot be computed directly, since neither the logging nor the candidate distributions have known density functions. We estimate them by training a multiclass domain classifier that learns to distinguish which policy generated a given impression~\citep{sugiyama2012density}. We construct a \emph{replay dataset} in which the logging class is represented by the actual logged impressions and each candidate class by replayed counterfactual actions on the same contexts. Counterfactual candidate actions $\pi_j(x_i)$ are produced by shadow-mode scoring in the production experiment and by offline replay of the candidate's scoring function on the logged contexts in simulation. Each impression is paired with a label $g \in \{0,1,\dots,N\}$ indicating its generating policy ($g{=}0$ for the logging policy, $g{=}j$ for candidate~$\pi_j$). We then train a domain classifier $q_\theta(g\mid z)$ to predict the policy label from the impression, using any architecture that outputs a probability distribution over the $N{+}1$ labels, with standard multiclass cross-entropy and uniform class priors.
}

\begin{definition}[Overlap and cold-start regions for policy $\pi_j$]
\label{def:overlap_coldstart}
For each policy class $g\in\{0,\ldots,N\}$, let $p_g(z)$ denote the density of $P_{\pi_g}$. Under equal class priors, let $q^\star(g\mid z)$ denote the Bayes, or population, domain posterior:
\[
q^\star(g\mid z)
\;\triangleq\;
\frac{p_g(z)}{\sum_{k=0}^{N} p_k(z)}.
\]
This is the true probability that impression $z$ was generated by policy class $g$ in the balanced domain-classification problem. Fix a threshold $\tau\in(0,1)$. For each candidate policy $\pi_j$, define
\begin{equation}
\begin{aligned}
\mathcal{O}_j(\tau)
&\;\triangleq\; \{\,z \in \mathrm{supp}(P_{\pi_j}) : q^\star(0\mid z) \ge \tau\,\}, \\
\mathcal{C}_j(\tau)
&\;\triangleq\; \{\,z \in \mathrm{supp}(P_{\pi_j}) : q^\star(0\mid z) < \tau\,\}.
\end{aligned}
\label{eq:overlap_coldstart_regions}
\end{equation}
and define the cold-start mass as $\mu_j(\tau)\triangleq P_{\pi_j}(\mathcal{C}_j(\tau))$. The learned domain classifier $q_\theta(g\mid z)$ provides plug-in estimates of these population quantities.
\end{definition}

\par{
Under equal class priors, the standard density-ratio identity~\citep{sugiyama2012density} gives
\begin{equation}
\frac{p_j(z)}{p_0(z)}
\;=\;
\frac{q^\star(j\mid z)}{q^\star(0\mid z)},
\qquad j=1,\dots,N.
\label{eq:domain_ratio_identity}
\end{equation}
Replacing the population posterior $q^\star$ with the learned classifier $q_\theta$ gives the estimated per-policy ratios below; substituting these estimates into~\eqref{eq:mixture_weight} gives the estimated mixture weight:
\begin{equation}
\hat\Omega_j(z) \;\triangleq\; \frac{q_\theta(j\mid z)}{q_\theta(0\mid z)},
\qquad
\hat\Omega(z) \;\triangleq\; \sum_{j=1}^{N} \beta_j\,\hat\Omega_j(z),
\end{equation}
where $\hat\Omega_j(z)$ is the \emph{estimated per-policy density ratio} of candidate $\pi_j$ to the logging policy.
}

\par{
The estimated weights $\hat\Omega(z)$ can have heavy tails when $q_\theta(0\mid z)$ is small, i.e., when the classifier is confident that an impression was \emph{not} generated by the logging policy.
To prevent a few extreme weights from dominating the training loss, we clip $\hat\Omega(z)$ to a maximum value~$c$:
\begin{equation}
\tilde\Omega(z) \;=\; \min\!\big(\hat\Omega(z),\; c\big).
\label{eq:omega_clip}
\end{equation}
We monitor the \emph{effective sample size}~\citep{kong1994sequential} (ESS) as a diagnostic for weight concentration:
\begin{equation}
\mathrm{ESS}
\;=\;
\frac{\big(\sum_{i=1}^n \tilde\Omega(z_i)\big)^2}{\sum_{i=1}^n \tilde\Omega(z_i)^2}.
\label{eq:ess}
\end{equation}
In our experiments, we report $\mathrm{ESS}/n$, which lies in $[1/n,1]$.
When $\mathrm{ESS}/n \approx 1$, the weights are roughly uniform and the logging and candidate distributions overlap well.
When $\mathrm{ESS}/n \ll 1$, a few impressions dominate the weighted loss, signaling poor overlap and potential instability in the value estimates.
The reliability of the estimated weights, and therefore of ESS as an overlap diagnostic, depends on whether the impression representation preserves policy-induced differences and whether the domain classifier produces reliable probabilities. We return to this issue empirically in Appendix~\ref{sec:dc_reliability_experiment}.
}

\subsection{Value Estimation Error Decomposition}
\label{sec:error_decomposition}

\par{
To understand when mixture-weighted training improves policy ranking and when it cannot, we bound the value estimation error by the sum of three terms: an overlap error, a cold-start term, and a replay concentration term.
Definition~\ref{def:overlap_coldstart} partitions the impressions produced by each candidate policy according to whether the logging policy provides high or low density there, with $\tau$ controlling a trade-off between the resulting error terms.
}

\begin{assumption}
\label{ass:consistency}
\renewcommand{\theenumi}{A\arabic{enumi}}
\renewcommand{\labelenumi}{(\theenumi)}
\begin{enumerate}
\setlength{\itemsep}{0pt}
\item\label{ass:A1} Rewards and predictions are bounded, $r^\star(z),\widehat r(z)\in[0,R_{\max}]$, and the training loss is bounded, $\ell\in[0,L_{\max}]$. For cross-entropy, reward probabilities are clipped to $[\varepsilon,1-\varepsilon]$, so $L_{\max}=-\log\varepsilon$; our runs use $\varepsilon=10^{-6}$.
\item\label{ass:A2} The loss is calibrated (proper). Its conditional excess risk is
\[
\Delta_\ell(f;z)
\triangleq
\mathbb E\!\left[
\ell(f(z),Y)-\ell(r^\star(z),Y)
\mid Z=z
\right],
\]
where $Y$ is the realized reward, and it satisfies
\[
|f(z)-r^\star(z)|^2
\le
c_\ell\Delta_\ell(f;z),
\]
with $c_\ell=1$ for squared loss and $c_\ell=1/2$ for log loss.
\item\label{ass:A3} Let $w(z)\equiv\Omega(z)$ denote the true mixture-to-logging density-ratio weight, let $\widehat w(z)\equiv\widehat\Omega(z)$ denote its unclipped domain-classifier estimate, and define $\widetilde w(z)\equiv\widetilde\Omega(z)=\min\{\widehat w(z),W\}$, where $W=c$ is the clipping threshold from~\eqref{eq:omega_clip}. Conditional on a separate domain-classifier fold, the clipped estimated weight $\widetilde w$ is fixed and satisfies
\[
0\le\widetilde w(z)\le W.
\]
Let $\mathcal D_{\mathrm{train}}\subset\mathcal D$ denote the reward-model training fold and let $n_{\mathrm{tr}}\triangleq|\mathcal D_{\mathrm{train}}|$ denote its sample size. Define its clipped weighted empirical risk by
\[
\widehat R_{\widetilde w}(f)
\triangleq
\frac{1}{n_{\mathrm{tr}}}
\sum_{(Z_i,Y_i)\in\mathcal D_{\mathrm{train}}}
\widetilde w(Z_i)\,\ell(f(Z_i),Y_i).
\]
The reward learner returns $\widehat r\in\mathcal F$ satisfying
\[
\widehat R_{\widetilde w}(\widehat r)
\le
\inf_{f\in\mathcal F}
\widehat R_{\widetilde w}(f)
+
\varepsilon_{\mathrm{opt}}.
\]
Here $\varepsilon_{\mathrm{opt}}\ge 0$ is the optimization error: the amount by which the clipped weighted empirical risk of $\widehat r$ exceeds its infimum over $\mathcal F$.
\item\label{ass:A4} The replay contexts $X_1,\ldots,X_m$ are i.i.d. draws from $P_X$ and are independent of the reward-model and domain-classifier training folds, and hence of the fitted reward model and estimated weights.
\end{enumerate}
\end{assumption}

\begin{theorem}[Replay-DM value error for a reward model trained with sample-split estimated mixture weights]
\label{thm:error_decomp}
Write the candidate-overlap error
\begin{equation}
\mathcal E_j^{\mathrm{ov}}(\tau)
\triangleq
\mathbb E_{P_{\pi_j}}
\left[
|\widehat r-r^\star|\,
\mathbf 1\{Z\in\mathcal O_j(\tau)\}
\right].
\label{eq:candidate_overlap_error}
\end{equation}
Under Assumption~\ref{ass:consistency}, for any $\delta_{\mathrm{tr}},\delta_{\mathrm{ev}}\in(0,1)$, with $\delta_{\mathrm{tr}}$ and $\delta_{\mathrm{ev}}$ denoting the failure probabilities of the training and replay-evaluation events, respectively, the following holds with probability at least $1-\delta_{\mathrm{tr}}-\delta_{\mathrm{ev}}$ over the joint training and replay samples, simultaneously for all candidates $j$:
\begin{equation}
\left|\widehat V_j^{\mathrm{DM}}-V_j\right|
\le
\mathcal E_j^{\mathrm{ov}}(\tau)
+
R_{\max}\mu_j(\tau)
+
R_{\max}
\sqrt{
\frac{\log(2N/\delta_{\mathrm{ev}})}{2m}
}.
\label{eq:error_bound_finite}
\end{equation}
On the same event, the overlap error admits two valid upper bounds, a change-of-measure bound and, under mixture-weighted training and the shared reward law of Section~\ref{sec:problem_setup}, a decomposition bound:
\begin{equation}
\mathcal E_j^{\mathrm{ov}}(\tau)
\le
\frac{1}{\tau}\epsilon_{0,j}(\widehat r;\tau)
\quad\text{(i)},
\qquad
\mathcal E_j^{\mathrm{ov}}(\tau)
\le
\sqrt{
\frac{c_\ell}{\beta_j}
\left(
A_{\mathcal F}
+\varepsilon_{\mathrm{opt}}
+2\Gamma_{n_{\mathrm{tr}}}(\delta_{\mathrm{tr}})
+2L_{\max}\varepsilon_w
\right)
}
\quad\text{(ii)}.
\label{eq:overlap_bounds}
\end{equation}
Thus $\mathcal E_j^{\mathrm{ov}}(\tau)\le\min\{\text{(i)},\text{(ii)}\}$, and neither bound dominates the other. Form (i) is a change of measure to the logging distribution that exposes the $1/\tau$ factor; form (ii) is the mixture-weighted training decomposition into model-fit, finite-sample, and weight terms.
\end{theorem}

\par{
The terms are as follows. Let $w(z)\equiv\Omega(z)$ denote the true mixture-to-logging density-ratio weight, let $\widehat w(z)\equiv\widehat\Omega(z)$ denote its unclipped domain-classifier estimate, and define $\widetilde w(z)\equiv\widetilde\Omega(z)=\min\{\widehat w(z),W\}$, where $W=c$ is the clipping threshold from~\eqref{eq:omega_clip}. Define the corresponding population risks
\begin{equation}
R_w(f)
\triangleq
\mathbb E_{P_{\pi_0}}\!\left[w(Z)\ell(f(Z),Y)\right],
\qquad
R_{\widetilde w}(f)
\triangleq
\mathbb E_{P_{\pi_0}}\!\left[\widetilde w(Z)\ell(f(Z),Y)\right].
\label{eq:population_weighted_risks}
\end{equation}
Let
\begin{equation}
f^\star_{\mathcal F}
\in
\argmin_{f\in\mathcal F}R_w(f)
\label{eq:best_in_class_predictor}
\end{equation}
denote the best predictor in the model class under the true weighted population risk. Define the logging-distribution overlap error by
\begin{equation}
\epsilon_{0,j}(\widehat r;\tau)
\triangleq
\mathbb E_{P_{\pi_0}}
\!\left[
|\widehat r-r^\star|\mathbf 1\{Z\in\mathcal O_j(\tau)\}
\right].
\label{eq:logging_overlap_error}
\end{equation}
Here, $\mathcal O_j(\tau)$ is the overlap region defined in Equation~\eqref{eq:overlap_coldstart_regions}. $A_{\mathcal F}=R_w(f^\star_{\mathcal F})-R_w(r^\star)$ is model-class approximation error (best-in-class weighted risk minus the Bayes risk). $\varepsilon_{\mathrm{opt}}$ is the optimization error defined in Assumption~\ref{ass:A3}. $\Gamma_{n_{\mathrm{tr}}}(\delta_{\mathrm{tr}})$ is a high-probability ($1-\delta_{\mathrm{tr}}$) bound on the weighted uniform deviation $\sup_{f\in\mathcal F}|R_{\widetilde w}(f)-\widehat R_{\widetilde w}(f)|$ (generalization) for the clipped-weighted loss class. Define
\begin{equation}
\varepsilon_{\mathrm{DC}}
\triangleq
\mathbb E_{P_{\pi_0}}\!\left[|\widehat w(Z)-w(Z)|\right],
\qquad
\varepsilon_{\mathrm{clip}}
\triangleq
\mathbb E_{P_{\pi_0}}\!\left[(w(Z)-W)_+\right],
\label{eq:dc_and_clipping_errors}
\end{equation}
and
\begin{equation}
\varepsilon_w
\triangleq
\mathbb E_{P_{\pi_0}}\!\left[|\widetilde w(Z)-w(Z)|\right]
\le
2\varepsilon_{\mathrm{DC}}+\varepsilon_{\mathrm{clip}}.
\label{eq:weight_error_bound}
\end{equation}
Thus $\varepsilon_w$ combines domain-classifier ratio-estimation error and clipping bias. $\mu_j(\tau)=P_{\pi_j}(\mathcal C_j(\tau))$ is the cold-start mass on the region $\mathcal C_j(\tau)$ defined in Equation~\eqref{eq:overlap_coldstart_regions}; the computable $q_\theta$ plug-in requires a classifier calibration or error correction. Finally, $\beta_j$ is the candidate mixture share, $P_{\mathrm{mix}}=\sum_j\beta_jP_{\pi_j}$; $c_\ell$ is the calibration constant ($1$ for squared loss and $1/2$ for log loss); and $m$ is the number of replay contexts.
}

\par{
The full proof is given in Appendix~\ref{app:error_analysis}.
}

\begin{remark}[$\tau$ trade-off and finite-sample selection]
\label{rem:tau_tradeoff}
Selecting $\tau$ is a trade-off. Decreasing $\tau$ lowers $\mu_j(\tau)$ but inflates the $1/\tau$ factor; increasing $\tau$ does the reverse. Finite-sample selection should minimize the worst-candidate full bound,
\begin{equation}
\tau^\star
\in
\argmin_{0<\tau<1}
\max_j
\left[
\min\left\{
\frac{\epsilon_{0,j}(\widehat r;\tau)}{\tau},
B_j
\right\}
+
R_{\max}\mu_j(\tau)
\right],
\label{eq:tau_selection}
\end{equation}
where $B_j$ is the $\tau$-independent mixture bound in form (ii) of Equation~\eqref{eq:overlap_bounds}, using estimable upper bounds/proxies since both branches depend on $r^\star$.
\end{remark}


\begin{corollary}[Top-1 selection regret]
\label{cor:regret}
Let $j^\star = \arg\max_{j} V_j$ denote the true best policy and
$\hat\jmath = \arg\max_{j} \widehat V_j^{\mathrm{DM}}$ the policy selected by offline evaluation.
Let $\epsilon_j = |\widehat V_j^{\mathrm{DM}} - V_j|$ denote the value estimation error for policy~$\pi_j$.
Then the simple regret of selecting $\hat\jmath$ satisfies
\begin{equation}
\mathrm{Regret}
\;\triangleq\;
V_{j^\star} - V_{\hat\jmath}
\;\le\;
2\,\max_{j=1,\dots,N} \epsilon_j,
\label{eq:regret_bound_simple}
\end{equation}
which follows directly from the triangle inequality applied to $\widehat V_{\hat\jmath}^{\mathrm{DM}} \ge \widehat V_{j^\star}^{\mathrm{DM}}$.
\end{corollary}

\par{
Corollary~\ref{cor:regret} shows that the regret of offline policy selection is controlled by the worst-case value estimation error across candidate policies. In particular, uniformly accurate value estimates guarantee reliable policy selection; large errors for either the optimal policy or the selected policy can directly translate into suboptimal deployment decisions.
The proof is a straightforward triangle-inequality argument and is deferred to Appendix~\ref{app:regret_proof}.
}

\begin{corollary}[Asymptotic offline-vs-A/B routing rule]
\label{cor:decision_rule}
Fix $N$ and $\tau$, and let $n_{\mathrm{tr}},m\to\infty$. Choose sequences $\delta_{\mathrm{tr}}(n_{\mathrm{tr}}),\delta_{\mathrm{ev}}(m)\in(0,1)$ such that $\delta_{\mathrm{tr}}(n_{\mathrm{tr}})\to0$, $\delta_{\mathrm{ev}}(m)\to0$, $\Gamma_{n_{\mathrm{tr}}}(\delta_{\mathrm{tr}}(n_{\mathrm{tr}}))\to0$, and $\log(2N/\delta_{\mathrm{ev}}(m))/m\to0$. Suppose also that $A_{\mathcal F}\to0$, $\varepsilon_{\mathrm{opt}}\to0$, and $\varepsilon_w\to0$, with convergence in probability for random terms. Then, for every $\eta>0$,
\begin{equation}
\Pr\!\left(
\mathrm{Regret}
>
2\,R_{\max}\,\max_j\mu_j(\tau)+\eta
\right)
\to 0.
\label{eq:asymptotic_regret_probability}
\end{equation}
Consequently, for a regret budget $\Delta>0$, if
\begin{equation}
2\,R_{\max}\,\max_j\mu_j(\tau)
\;\le\;
\Delta,
\label{eq:decision_rule}
\end{equation}
then $\Pr(\mathrm{Regret}>\Delta+\eta)\to0$ for every $\eta>0$. In this asymptotic regime, the routing rule trusts the offline-selected policy when~\eqref{eq:decision_rule} holds and routes the candidate set to A/B testing otherwise.
\end{corollary}

\par{
Corollary~\ref{cor:decision_rule} is a population asymptotic routing rule, applicable when the reducible model-fit, weighted-generalization, weight, and replay finite-sample terms in Theorem~\ref{thm:error_decomp} are negligible. Its $q_\theta$-based implementation additionally requires classifier calibration/error control. For finite samples, the routing decision must use the full value-error bound in Theorem~\ref{thm:error_decomp}. The proof of Corollary~\ref{cor:decision_rule} is given in Appendix~\ref{app:decision_rule_proof}.
}

\section{Synthetic Data Experiments}
\label{sec:synthetic_experiments}

\subsection{Simulation Environment and Setup}
\label{sec:sim_setup}

\par{
We developed a custom simulation framework for offline ranking policy evaluation. The simulator maintains a fixed corpus of users and ads, each represented as a feature vector, and an oracle reward function that captures nonlinear user-ad interactions to assign an expected reward to any (user, ad) pair. Each logged impression consists of a randomly drawn user presented with $k$ candidate ads; the acting policy scores each candidate and selects the highest-scoring ad, and the oracle evaluates the resulting (user, ad) pair to produce a reward. We use $10{,}000$ users, $1{,}000$ ads, and $k=10$ candidates per impression, with $70{,}000$ logged impressions per trial split into four disjoint folds: $20{,}000$ for domain-classifier training, $20{,}000$ for reward-model training, $10{,}000$ for calibration, and $20{,}000$ for evaluation. Unless otherwise noted, reward models are trained for $30$ epochs, domain classifiers for $4$ epochs, estimated weights are clipped at $c=100$, and results are aggregated over $50$ seeds. We publicly release the simulator and synthetic evaluation code.\footnote{\url{https://github.com/facebookresearch/ad_ranking_simulation}} The simulator provides full observability: exact ground-truth policy values are obtained by applying each candidate to the held-out evaluation contexts and querying the oracle. Each candidate policy scores (user, ad) pairs the same way as the logging policy but with shifted scoring parameters. A calibration sweep over these shifts places the candidates of a trial at value multipliers $V_j/V_0\in\{0.95,1.00,\dots,1.20\}$ and at a common target ESS of their true density ratios relative to the logging policy, one of $0.25,0.35,\dots,0.85$; if no shift reaches a candidate's value target, that candidate is omitted. Full details of feature distributions, oracle reward function, logging-policy construction, and the calibration sweep are in Appendix~\ref{app:simulator_details}.
}

\subsection{Estimators}
\label{sec:estimators}

\par{
We evaluate six estimators in the synthetic experiments. Let $\{(z_i,r_i)\}_{i=1}^m$ denote the logged impressions in the evaluation split, $\widehat r(z)$ a learned reward model, and $\widehat\Omega_j(z)$ the estimated per-policy density-ratio weight from the multiclass domain classifier of Section~\ref{sec:domain_classifier}. Define its clipped version as $\tilde\Omega_j(z)\triangleq\min\{\widehat\Omega_j(z),c\}$.
(i) \textbf{IPS}: importance-weighted average of logged rewards, $\widehat V_j^{\mathrm{IPS}}=(1/m)\sum_i\tilde\Omega_j(z_i)r_i$. Since true propensities are unavailable, IPS uses the same classifier-based density-ratio estimates as the weighted DM and doubly robust (DR) estimators.
(ii) \textbf{Unweighted DM}: the standard direct method with $\widehat r$ trained on the logging distribution under uniform sample weights and replayed on candidate-policy actions.
(iii) \textbf{Weighted DM (ours)}: the mixture-weighted reward model from Section~\ref{sec:mixture_weighted}, trained with the clipped mixture weight $\tilde\Omega(z)$ from Equation~\eqref{eq:omega_clip} and replayed on candidate-policy actions.
(iv) \textbf{Per-policy Weighted DM}: $N$ separately trained reward models, with model $j$ trained using $\tilde\Omega_j(z)$; it uses the same data splits, architecture, optimizer, epochs, calibration, and clipping settings as Weighted DM.
(v) \textbf{Unweighted DR}: the doubly robust method using the unweighted reward model plus the residual correction $(1/m)\sum_i\tilde\Omega_j(z_i)(r_i-\widehat r(z_i))$.
(vi) \textbf{Weighted DR (ours)}: the same residual correction using the mixture-weighted reward model as the control variate.
}

\subsection{Evaluation Metrics}
\label{sec:metrics}

\par{
In each trial, we evaluate $N$ candidate target policies $\pi_1, \ldots, \pi_N$. Let $V_j \triangleq V(\pi_j)$ denote the oracle value of $\pi_j$, and let $\widehat{V}_j$ denote the policy-value estimate produced by each method. We report the following metrics. \textbf{Top-1 regret (Reg).} Let $\hat{\jmath} \triangleq \arg\max_{j} \widehat V_j$ be the deployed policy choice and $j^\star \triangleq \arg\max_{j} V_j$ the oracle-best. Top-1 regret is $V_{j^\star}-V_{\hat{\jmath}}$ averaged over trials: the value lost by deploying the offline-selected policy instead of the true best. For Per-policy Weighted DM, whose $N$ reward models each produce their own estimate $\widehat V_0^{(j)}$ of the logging-policy value, $\hat\jmath$ maximizes the estimated improvement $\widehat V_j-\widehat V_0^{(j)}$. \textbf{Top-1 accuracy (Acc).} $\mathbb{I}\{\hat{\jmath}=j^\star\}$ averaged over trials: the fraction of trials in which the correct policy is selected. \textbf{Value MAE\% (MAE\%).} $\frac{100}{N}\sum_{j=1}^{N} \left|\widehat V_j - V_j\right|/V_j$: the mean absolute percentage error of the estimated policy values $\widehat V_j$ against the oracle values $V_j$, averaged over the $N$ candidate policies in a trial.
}

\subsection{Experiment 1: Estimator Comparison Under Distribution Shift}
\label{sec:exp1_baseline_domain_shift}

\par{
We compare the six estimators of Section~\ref{sec:estimators} across the seven ESS targets of Section~\ref{sec:sim_setup}. With $50$ seeds per target, this gives $350$ trials, grouped into three regimes by ESS target (Table~\ref{tab:method_comparison_regimes}). The regime headers in Table~\ref{tab:method_comparison_regimes} report the corresponding mean cold-start mass at $\tau=0.1$, linking the empirical ESS grouping to the overlap diagnostic in the value-error bound.}

\par{
Table~\ref{tab:method_comparison_regimes} reports performance across the three ESS regimes. The IPS baseline ranges from $0.60$ to $0.71$ Top-1 accuracy across regimes. Mixture weighting clearly helps: Weighted DM improves over Unweighted DM on eight of nine regime-metric comparisons, and it outperforms Per-policy Weighted DM on all nine, so its gains cannot be explained by candidate-specific reweighting alone. Weighted DR is best on all three metrics in the Low and Mid regimes; in the High regime, Unweighted DR has slightly higher Top-1 accuracy, the two DR estimators tie on regret, and Weighted DR has lower Value MAE\%. Full per-bin breakdowns with standard errors and the per-policy Value MAE\% breakdown are in Appendix~\ref{app:detailed_results}. To test sensitivity to the clip $c$ in Equation~\eqref{eq:omega_clip}, which trades off removing the up-weighting that corrects the shift (small $c$) against letting a few impressions dominate the training loss (large $c$), we repeat the Experiment~1 setup with $c\in\{1,5,10,20,100\}$; the results are in Appendix~\ref{app:exp2_clipping}. At every overlap level, Value MAE\% is nearly the same for $c\in\{10,20,100\}$, a range that includes the $c=100$ used in Experiment~1. In contrast, $c=1$, which removes all up-weighting, stays close to the unweighted model at low overlap and is best only at high overlap.
}

\begin{table*}[t]
\centering
\caption{Comparison of six policy-selection methods across ESS regimes in the $70{,}000$-impression sample-split experiment (50 seeds, up to 6 candidate policies $\times$ 7 ESS targets per seed). Column headers report the mean cold-start mass at $\tau=0.1$ in each regime, estimated by thresholding the learned classifier $q_\theta$ in place of $q^\star$ in Definition~\ref{def:overlap_coldstart}. Classifier-weighted variants and IPS use clip $c=100$. Best per column in \textbf{bold}.}
\label{tab:method_comparison_regimes}
\small
\setlength{\tabcolsep}{3pt}
\renewcommand{\arraystretch}{1.05}
\begin{tabular}{l ccc ccc ccc}
\toprule
& \multicolumn{3}{c}{\shortstack{Low ESS ($<0.4$) \\ $\bar P_{\pi_j}(\mathcal{C}_j(0.1)){\approx}\,43\%$}}
& \multicolumn{3}{c}{\shortstack{Mid ESS ($0.4$--$0.7$) \\ $\bar P_{\pi_j}(\mathcal{C}_j(0.1)){\approx}\,25\%$}}
& \multicolumn{3}{c}{\shortstack{High ESS ($\geq 0.7$) \\ $\bar P_{\pi_j}(\mathcal{C}_j(0.1)){\approx}\,12\%$}} \\
\cmidrule(lr){2-4}\cmidrule(lr){5-7}\cmidrule(lr){8-10}
Method & Acc$\uparrow$ & Reg$\downarrow$ & MAE\%$\downarrow$ & Acc$\uparrow$ & Reg$\downarrow$ & MAE\%$\downarrow$ & Acc$\uparrow$ & Reg$\downarrow$ & MAE\%$\downarrow$ \\
\midrule
IPS                            & 0.60 & 0.0041 & 7.64 & 0.65 & 0.0045 & 5.49 & 0.71 & 0.0028 & 4.15 \\
\midrule
Unweighted DM                  & 0.67 & 0.0044 & 7.20 & 0.87 & 0.0006 & 2.58 & 0.96 & 0.0003 & 1.59 \\
Weighted DM (ours)             & 0.81 & 0.0014 & 2.65 & 0.91 & 0.0004 & 1.94 & 0.98 & \textbf{0.0000} & 1.86 \\
Per-policy Weighted DM         & 0.75 & 0.0028 & 3.61 & 0.83 & 0.0013 & 2.80 & 0.90 & 0.0006 & 2.32 \\
\midrule
Unweighted DR                  & 0.82 & 0.0016 & 2.64 & 0.93 & 0.0002 & 0.70 & \textbf{1.00} & \textbf{0.0000} & 0.30 \\
Weighted DR (ours)             & \textbf{0.90} & \textbf{0.0005} & \textbf{1.71} & \textbf{0.95} & \textbf{0.0001} & \textbf{0.43} & 0.99 & \textbf{0.0000} & \textbf{0.21} \\
\bottomrule
\end{tabular}
\end{table*}

\section{Production Data Experiment}
\label{sec:exp_prod}

\par{
\textbf{Setup.}
We evaluate production traffic from a large-scale industrial CTR ranking system with one production logging policy and five candidate treatment policies. Each of the six policies served a separate share of live production traffic in an online A/B test, so its true value $V_j$ is the observed click-through rate on the evaluation impressions it served. Top-1 accuracy is computed over the five candidates, so random selection gives an accuracy of $1/5=0.20$. Value MAE\% averages over all six policies, including the logging policy. All metrics are averaged over $50$ random subsamples of the production data, with estimated weights clipped at $c=20$. Weighted DM and Per-policy Weighted DM use the same data splits, architectures, optimization, calibration, and clipping settings.
}

\begin{table*}[t]
\centering
\caption{Production experiment on a large-scale industrial CTR ranking system. All metrics are averaged over $50$ random subsamples, with estimated weights clipped at $c=20$. Horizontal rules separate IPS, DM, and DR estimators; bold marks the best value per column within the DM and DR blocks.}
\label{tab:prod_results}
\small
\setlength{\tabcolsep}{3pt}
\renewcommand{\arraystretch}{1.05}
\begin{tabular}{l ccc}
\toprule
Method & Acc$\uparrow$ & Reg$\downarrow$ & MAE\%$\downarrow$ \\
\midrule
IPS & 0.24 & 0.00807 & 4.14 \\
\midrule
Unweighted DM & 0.36 & 0.00271 & 14.10 \\
Weighted DM (ours) & \textbf{0.40} & \textbf{0.00244} & \textbf{11.56} \\
Per-policy Weighted DM & 0.24 & 0.00676 & 12.43 \\
\midrule
Unweighted DR & 0.38 & \textbf{0.00225} & \textbf{1.90} \\
Weighted DR (ours) & \textbf{0.44} & 0.00226 & 2.04 \\
\bottomrule
\end{tabular}
\end{table*}

\par{
\textbf{Results.}
Among the DM estimators in Table~\ref{tab:prod_results}, Weighted DM performs best on all three metrics. Relative to the Per-policy Weighted DM baseline, it improves Top-1 accuracy from $0.24$ to $0.40$, reduces regret from $0.00676$ to $0.00244$, and reduces Value MAE\% from $12.43$ to $11.56$. It also outperforms Unweighted DM, which obtains $0.36$ Top-1 accuracy, $0.00271$ regret, and a Value MAE\% of $14.10$. Per-policy Weighted DM reweights toward each candidate with its own reward model, so the gain of Weighted DM cannot be explained by per-candidate reweighting or by training more reward models.
}

\par{
Among the DR estimators, Weighted DR achieves the highest Top-1 accuracy, while Unweighted DR has marginally lower regret and Value MAE\%, so neither DR variant dominates. IPS and both DR estimators attain much lower Value MAE\% than all DM estimators ($1.90$--$4.14$ versus $11.56$--$14.10$), but IPS has the highest regret ($0.00807$) and ties Per-policy Weighted DM for the lowest Top-1 accuracy ($0.24$). The production candidates are nearly tied: Weighted DM's estimated winner exceeds the runner-up by only $0.15\%$ on average (descriptive paired-bootstrap 95\% interval $[-0.31\%, 0.60\%]$), and their ordering reverses in $24$ of the $50$ subsamples, so offline evaluation cannot reliably separate the top two candidates, and this candidate set should be sent to online A/B testing. Per-policy Value MAE\% for each estimator is in Appendix~\ref{app:prod_detailed}.
}

\section{Conclusions}
\label{sec:conclusions}

\par{
We studied offline evaluation of ad-ranking policies in multi-stage production pipelines, where logged propensities are unavailable and direct-method (DM) replay with learned reward models is often the only viable option. We proposed mixture-weighted reward modeling, which trains the reward model on logged data reweighted toward a mixture of candidate-policy distributions using density ratios estimated by a multiclass domain classifier. We bounded each candidate's replay-DM value error by an overlap term, a cold-start term, and a replay concentration term, and derived an asymptotic rule that uses the cold-start term to decide, for a given regret budget, whether to trust the offline selection or route the candidates to online A/B testing.
}

\par{
In both the synthetic and production experiments, training a single reward model toward the candidate mixture makes DM replay more reliable under policy-induced shift: it improves policy selection most where overlap is weakest, and with a doubly robust correction it gives the highest Top-1 accuracy at low and medium overlap and on production data. The analysis also identifies the error that more logged data does not reduce: the bound retains a term for impressions a candidate generates where the logging data provide no support, and the cold-start mass bounds their share. Offline evaluation in multi-stage ad systems should therefore pair mixture-weighted reward models with an overlap diagnostic, trusting offline selection where overlap is adequate and reserving online A/B tests for candidates with large cold-start mass.
}

\par{
\textbf{Limitations.}
Our analysis assumes candidate policies can be replayed on logged contexts and that an impression's reward distribution does not depend on the policy that produced it. The cold-start mass is defined by the population domain posterior; computing it with the learned classifier requires calibration or error correction, and the weights are useful only when the representation retains policy-induced differences (Appendix~\ref{sec:dc_reliability_experiment}). The routing rule is asymptotic, and its operational performance under live traffic is left to future work. It also assumes that the weight error $\varepsilon_w$ vanishes, which, with a fixed clip $c$, requires the mixture ratio $\Omega(z)$ never to exceed $c$.
}

\begin{ack}
This work was funded by Meta Platforms, Inc. All authors are employees of Meta Platforms, Inc.
\end{ack}

\bibliography{references.bib}
\bibliographystyle{plainnat}

\newpage
\appendix

\section{Synthetic Simulator Details}
\label{app:simulator_details}

\par{
This appendix gives the simulator design used in Section~\ref{sec:synthetic_experiments}: the user and ad features, the oracle reward function, the logging- and candidate-policy parameterization, and the per-trial pipeline.
}

\subsection{Users, ads, and impression generation}
\par{
The simulator maintains a fixed corpus of $10{,}000$ users and $1{,}000$ ads, each represented by a feature vector drawn i.i.d.\ uniformly from $[-1,1]^3$. Each logged impression is generated as follows: a user is sampled uniformly at random from the corpus, $k=10$ candidate ads are sampled uniformly at random with replacement from the ad corpus, the acting policy scores each (user, ad) pair, and the highest-scoring ad is selected. The reward of an impression is the oracle's expected reward for the resulting (user, ad) pair, without Bernoulli sampling. Each trial uses $70{,}000$ impressions generated in this fashion.
}

\subsection{Oracle reward function}\label{app:oracle}
\par{
The oracle is a sum of $400$ Gaussian kernels in the joint user--ad feature space, each the product of a kernel on the user features and a kernel on the ad features, truncated at $1$:
\[
r^\star(x,a)=\min\Big\{1,\ \sum_{l=1}^{400}\alpha_l\exp\!\Big(-\frac{\|x-\mu^{x}_l\|^2}{2(\sigma^{x}_l)^2}\Big)\exp\!\Big(-\frac{\|a-\mu^{a}_l\|^2}{2(\sigma^{a}_l)^2}\Big)\Big\}.
\]
Here $x,a\in[-1,1]^3$ are the user and ad feature vectors. The kernel centers $\mu^{x}_l,\mu^{a}_l$ are drawn uniformly within $\pm0.5$ of the origin in each coordinate, widths $\sigma^{x}_l,\sigma^{a}_l$ uniformly from $[0.3,0.5]$, and amplitudes $\alpha_l$ uniformly from $[0.2,1.0]$, then rescaled so that the average reward of users paired with uniformly random ads is $0.15$. The oracle is redrawn with each random seed, so all candidate policies in a trial are scored against the same ground truth. The oracle value $V_j$ of a policy $\pi_j$ is obtained by applying $\pi_j$ to the evaluation contexts and averaging the oracle reward of the ads it selects.
}

\subsection{Logging policy and candidate-policy parameterization}
\par{
The logging policy $\pi_0$ and the candidate policies $\pi_1,\dots,\pi_6$ score a (user, ad) pair with a sum of $200$ Gaussian kernels of the same product form as the oracle's (Appendix~\ref{app:oracle}), whose centers lie within $\pm0.3$ of a policy center in each coordinate, with widths drawn uniformly from $[0.3,0.5]$ and amplitudes from $[0.1,1.0]$. The logging-policy center is drawn uniformly from $[-0.5,0.5]^6$ for each seed. Candidate centers are positioned by a calibration sweep over a grid of offsets along the line from the oracle center (the origin) to the logging-policy center and along a random orthogonal direction. For each grid point, the sweep computes the value multiplier $V_j/V_0$ and the ESS of the true density ratio relative to $\pi_0$, computed from the kernel scores, on $10{,}000$ logged impressions. For each of the six target multipliers $v_k\in\{0.95,1.00,1.05,1.10,1.15,1.20\}$ and each of the seven ESS targets $\{0.25,0.35,\dots,0.85\}$, the grid point whose multiplier is within $0.03$ of $v_k$ (or within $0.06$ if none is) and whose ESS is closest to the target is selected, and the candidate is built at those offsets with its own random orthogonal direction and kernel draw. If no grid point qualifies, the candidate is omitted; as a result, $56$ of the $350$ trials evaluate fewer than six candidates.
}

\subsection{Per-trial procedure}
\label{app:per_trial_procedure}
\par{
Each trial follows the same pipeline:\\
(1)~\textbf{Data collection and split.}
$70{,}000$ logged impressions are collected under $\pi_0$ and partitioned into four disjoint folds: $20{,}000$ for domain-classifier training, $20{,}000$ for reward-model training, $10{,}000$ for calibration, and $20{,}000$ for evaluation.\\
(2)~\textbf{Domain classifier training and calibration.}
Each candidate policy $\pi_j$ is replayed on the contexts of the domain-classifier fold to produce counterfactual impressions; these, together with that fold's logged impressions, form a multiclass dataset labeled by the generating policy. A one-hidden-layer MLP with $64$ units is trained on this dataset with Adam (learning rate $0.01$, batch size $512$) for $4$ epochs and calibrated by temperature scaling fitted on 5-fold out-of-fold logits.\\
(3)~\textbf{Importance weights.}
The calibrated classifier is applied to the reward-model fold, and its outputs are converted to the clipped mixture weight $\tilde\Omega(z)$, with uniform $\beta_j=1/N$, and the clipped per-policy weights $\tilde\Omega_j(z)$, using $c=100$. Value MAE\% is nearly the same for $c\in\{10,20,100\}$, a range that includes this $c=100$ (Appendix~\ref{app:exp2_clipping}).\\
(4)~\textbf{Reward model training.}
The unweighted, mixture-weighted, and per-policy weighted reward models share the same architecture, a one-hidden-layer MLP with $64$ units, and are trained on the reward-model fold with Adam (learning rate $10^{-3}$, batch size $512$) for $30$ epochs under a weighted log loss. They differ only in the sample weights: uniform, $\tilde\Omega(z)$, or $\tilde\Omega_j(z)$. Each model is temperature-calibrated on the calibration fold.\\
(5)~\textbf{Evaluation.}
On the evaluation fold, we compute the six estimators of Section~\ref{sec:estimators} for each candidate policy, with IPS and DR using the clipped per-policy weights $\tilde\Omega_j$ of the evaluation fold's logged impressions. Oracle values $V_j$ are obtained by applying each policy to the evaluation contexts and querying the oracle. The three metrics of Section~\ref{sec:metrics} are computed for each estimator. All training runs use a CPU.
}

\section{Proof of Theorem~\ref{thm:error_decomp}: Replay-DM Value Estimation Error}
\label{app:error_analysis}

\begin{proof}
Write $e(z)=\widehat r(z)-r^\star(z)$.

\paragraph{Replay decomposition.}
Since $Z=(X,\pi_j(X))$ with $X\sim P_X$ has distribution $P_{\pi_j}$ (Section~\ref{sec:problem_setup}), Equation~\eqref{eq:policy_value} gives $V_j=\mathbb E_{P_{\pi_j}}[r^\star(Z)]$. Adding and subtracting the population replay mean $\mathbb E_{P_{\pi_j}}[\widehat r(Z)]$ gives
\[
\widehat V_j^{\mathrm{DM}}-V_j
=
\Big(\frac{1}{m}\sum_{s=1}^m \widehat r\big(X_s,\pi_j(X_s)\big)-\mathbb E_{P_{\pi_j}}[\widehat r(Z)]\Big)
+\mathbb E_{P_{\pi_j}}[e(Z)],
\]
and the triangle inequality yields
\begin{align}
\left|\widehat V_j^{\mathrm{DM}}-V_j\right|
&\le
\left|
\frac{1}{m}\sum_{s=1}^m
\widehat r\!\left(X_s,\pi_j(X_s)\right)
-
\mathbb E_{P_{\pi_j}}[\widehat r(Z)]
\right|
+
\left|\mathbb E_{P_{\pi_j}}[e(Z)]\right|.
\label{eq:proof_replay_decomposition}
\end{align}
Because $\widehat r$ is trained only on the training folds, it is a fixed function once we condition on them. By Assumption~\ref{ass:A4}, the replay contexts $X_1,\ldots,X_m$ are i.i.d.\ from $P_X$ and independent of those folds, so the summands $\widehat r(X_s,\pi_j(X_s))$ are i.i.d.\ with mean $\mathbb E_{P_{\pi_j}}[\widehat r(Z)]$ and lie in $[0,R_{\max}]$ by Assumption~\ref{ass:A1}. Hoeffding's inequality followed by a union bound over the $N$ candidates therefore gives, with probability at least $1-\delta_{\mathrm{ev}}$,
\begin{equation}
\left|
\frac{1}{m}\sum_{s=1}^m
\widehat r\!\left(X_s,\pi_j(X_s)\right)
-
\mathbb E_{P_{\pi_j}}[\widehat r(Z)]
\right|
\le
R_{\max}\sqrt{\frac{\log(2N/\delta_{\mathrm{ev}})}{2m}}
\label{eq:proof_replay_concentration}
\end{equation}
simultaneously for all $j$.
Splitting the population bias over the regions in Equation~\eqref{eq:overlap_coldstart_regions} yields
\begin{align}
\left|\mathbb E_{P_{\pi_j}}[e(Z)]\right|
&\le
\mathbb E_{P_{\pi_j}}
\!\left[
|e(Z)|\mathbf 1\{Z\in\mathcal O_j(\tau)\}
\right]
+
\mathbb E_{P_{\pi_j}}
\!\left[
|e(Z)|\mathbf 1\{Z\in\mathcal C_j(\tau)\}
\right]
\nonumber\\
&\le
\mathcal E_j^{\mathrm{ov}}(\tau)
+
R_{\max}\mu_j(\tau),
\label{eq:proof_overlap_coldstart_split}
\end{align}
where the second inequality uses $|e(z)|\le R_{\max}$ from Assumption~\ref{ass:A1}. Combining Equations~\eqref{eq:proof_replay_decomposition}--\eqref{eq:proof_overlap_coldstart_split} proves Equation~\eqref{eq:error_bound_finite}.

\paragraph{Change-of-measure bound (i).}
For $z\in\mathcal O_j(\tau)$, Definition~\ref{def:overlap_coldstart} implies
\[
q^\star(0\mid z)
=
\frac{p_0(z)}{\sum_{k=0}^N p_k(z)}
\ge \tau,
\]
so $p_0(z)\ge\tau p_j(z)$ and hence $p_j(z)/p_0(z)\le1/\tau$. Therefore,
\begin{align}
\mathcal E_j^{\mathrm{ov}}(\tau)
&=
\mathbb E_{P_{\pi_0}}
\!\left[
\frac{p_j(Z)}{p_0(Z)}
|e(Z)|\mathbf 1\{Z\in\mathcal O_j(\tau)\}
\right]
\nonumber\\
&\le
\frac{1}{\tau}
\mathbb E_{P_{\pi_0}}
\!\left[
|e(Z)|\mathbf 1\{Z\in\mathcal O_j(\tau)\}
\right]
=
\frac{1}{\tau}\epsilon_{0,j}(\widehat r;\tau),
\label{eq:proof_change_of_measure}
\end{align}
which is bound (i) in Equation~\eqref{eq:overlap_bounds}.

\paragraph{Mixture-weighted bound (ii).}
Because $P_{\mathrm{mix}}=\sum_{k=1}^N\beta_kP_{\pi_k}$ with $\beta_j>0$, its density satisfies $p_{\mathrm{mix}}(z)\ge\beta_jp_j(z)$. Cauchy--Schwarz and the fact that $\mathcal O_j(\tau)$ lies on logging support give
\begin{align}
\left(\mathcal E_j^{\mathrm{ov}}(\tau)\right)^2
&\le
\mathbb E_{P_{\pi_j}}
\!\left[
e(Z)^2\mathbf 1\{Z\in\mathcal O_j(\tau)\}
\right]
\nonumber\\
&\le
\frac{1}{\beta_j}
\mathbb E_{P_{\mathrm{mix}}}
\!\left[
e(Z)^2\mathbf 1\{p_0(Z)>0\}
\right]
\nonumber\\
&\le
\frac{c_\ell}{\beta_j}
\left(R_w(\widehat r)-R_w(r^\star)\right).
\label{eq:proof_candidate_to_weighted_risk}
\end{align}
The last inequality applies the calibration condition of Assumption~\ref{ass:A2}; the identification with weighted excess risk uses the shared conditional reward law given $Z$ and the support-restricted mixture identity in Equation~\eqref{eq:weighted_loss}.

It remains to control the weighted excess risk. Let $f^\star_{\mathcal F}$ be the best-in-class predictor from Equation~\eqref{eq:best_in_class_predictor}. Adding and subtracting the clipped-weight population and empirical risks gives
\begin{align}
R_w(\widehat r)-R_w(f^\star_{\mathcal F})
&=
\bigl(R_w(\widehat r)-R_{\widetilde w}(\widehat r)\bigr)
+
\bigl(R_{\widetilde w}(\widehat r)-\widehat R_{\widetilde w}(\widehat r)\bigr)
\nonumber\\
&\quad+
\bigl(\widehat R_{\widetilde w}(\widehat r)-\widehat R_{\widetilde w}(f^\star_{\mathcal F})\bigr)
+
\bigl(\widehat R_{\widetilde w}(f^\star_{\mathcal F})-R_{\widetilde w}(f^\star_{\mathcal F})\bigr)
\nonumber\\
&\quad+
\bigl(R_{\widetilde w}(f^\star_{\mathcal F})-R_w(f^\star_{\mathcal F})\bigr).
\label{eq:proof_weighted_erm_telescope}
\end{align}
On the training event, the second and fourth terms are each at most $\Gamma_{n_{\mathrm{tr}}}(\delta_{\mathrm{tr}})$ in absolute value, while Assumption~\ref{ass:A3} bounds the middle empirical-risk difference by $\varepsilon_{\mathrm{opt}}$. For any $f\in\mathcal F$,
\begin{equation}
\left|R_w(f)-R_{\widetilde w}(f)\right|
\le
L_{\max}\mathbb E_{P_{\pi_0}}[|w(Z)-\widetilde w(Z)|]
=
L_{\max}\varepsilon_w.
\label{eq:proof_weight_risk_difference}
\end{equation}
Applying this inequality to the first and fifth terms of Equation~\eqref{eq:proof_weighted_erm_telescope}, then adding the approximation error $A_{\mathcal F}=R_w(f^\star_{\mathcal F})-R_w(r^\star)$, gives
\begin{equation}
R_w(\widehat r)-R_w(r^\star)
\le
A_{\mathcal F}
+
\varepsilon_{\mathrm{opt}}
+
2\Gamma_{n_{\mathrm{tr}}}(\delta_{\mathrm{tr}})
+
2L_{\max}\varepsilon_w.
\label{eq:proof_weighted_excess_risk}
\end{equation}
Substituting Equation~\eqref{eq:proof_weighted_excess_risk} into Equation~\eqref{eq:proof_candidate_to_weighted_risk} and taking square roots proves bound (ii) in Equation~\eqref{eq:overlap_bounds}.

\paragraph{Joint probability.}
The weighted-risk bound holds on the training event with probability at least $1-\delta_{\mathrm{tr}}$, and the simultaneous replay bound holds with probability at least $1-\delta_{\mathrm{ev}}$. A final union bound gives joint probability at least $1-\delta_{\mathrm{tr}}-\delta_{\mathrm{ev}}$, completing the proof.
\end{proof}

\section{Proof of Corollary~\ref{cor:regret}: Top-1 Selection Regret}
\label{app:regret_proof}

\begin{proof}
By definition of $\hat\jmath$, we have $\widehat V_{\hat\jmath}^{\mathrm{DM}} \ge \widehat V_{j^\star}^{\mathrm{DM}}$.
We decompose the regret as
\begin{align}
V_{j^\star} - V_{\hat\jmath}
&=
(V_{j^\star} - \widehat V_{j^\star}^{\mathrm{DM}})
+
(\widehat V_{j^\star}^{\mathrm{DM}} - \widehat V_{\hat\jmath}^{\mathrm{DM}})
+
(\widehat V_{\hat\jmath}^{\mathrm{DM}} - V_{\hat\jmath}).
\nonumber
\end{align}
The middle term is non-positive since $\widehat V_{\hat\jmath}^{\mathrm{DM}} \ge \widehat V_{j^\star}^{\mathrm{DM}}$, and the remaining terms can be bounded using absolute values:
\begin{align}
V_{j^\star} - V_{\hat\jmath}
&\le
|V_{j^\star} - \widehat V_{j^\star}^{\mathrm{DM}}|
+
|\widehat V_{\hat\jmath}^{\mathrm{DM}} - V_{\hat\jmath}|.
\nonumber
\end{align}
Finally, each term is bounded by the maximum estimation error across all policies, yielding
\begin{align}
V_{j^\star} - V_{\hat\jmath}
\le
2\,\max_{j=1,\dots,N} \epsilon_j.
\nonumber
\end{align}
\end{proof}

\section{Proof of Corollary~\ref{cor:decision_rule}: Asymptotic Offline-vs-A/B Routing Rule}
\label{app:decision_rule_proof}

\begin{proof}
Write $\delta_{\mathrm{tr}}=\delta_{\mathrm{tr}}(n_{\mathrm{tr}})$ and $\delta_{\mathrm{ev}}=\delta_{\mathrm{ev}}(m)$, and let $B_j$ denote the right-hand side of form (ii) in Equation~\eqref{eq:overlap_bounds},
\[
B_j
=
\sqrt{
\frac{c_\ell}{\beta_j}
\left(
A_{\mathcal F}
+\varepsilon_{\mathrm{opt}}
+2\Gamma_{n_{\mathrm{tr}}}(\delta_{\mathrm{tr}})
+2L_{\max}\varepsilon_w
\right)
}.
\]
By Theorem~\ref{thm:error_decomp}, with probability at least $1-\delta_{\mathrm{tr}}-\delta_{\mathrm{ev}}$, simultaneously for all candidates $j$,
\[
\left|\widehat V_j^{\mathrm{DM}}-V_j\right|
\le
B_j
+
R_{\max}\mu_j(\tau)
+
R_{\max}\sqrt{\frac{\log(2N/\delta_{\mathrm{ev}})}{2m}}.
\]
On the same event, Corollary~\ref{cor:regret} gives
\[
\mathrm{Regret}
\le
2\max_{j}\left|\widehat V_j^{\mathrm{DM}}-V_j\right|
\le
2R_{\max}\max_j\mu_j(\tau)+\rho,
\]
where
\[
\rho
\triangleq
2\max_j B_j
+
2R_{\max}\sqrt{\frac{\log(2N/\delta_{\mathrm{ev}})}{2m}}.
\]
Because $N$ and $\beta_1,\dots,\beta_N>0$ are fixed, the assumed limits $A_{\mathcal F}\to0$, $\varepsilon_{\mathrm{opt}}\to0$, $\Gamma_{n_{\mathrm{tr}}}(\delta_{\mathrm{tr}})\to0$, and $\varepsilon_w\to0$ give $\max_j B_j\to0$ in probability, and $\log(2N/\delta_{\mathrm{ev}})/m\to0$ sends the replay term to zero, so $\rho\to0$ in probability. Hence, for every $\eta>0$,
\[
\Pr\!\left(
\mathrm{Regret}
>
2R_{\max}\max_j\mu_j(\tau)+\eta
\right)
\le
\delta_{\mathrm{tr}}+\delta_{\mathrm{ev}}+\Pr(\rho>\eta)
\to0,
\]
which is Equation~\eqref{eq:asymptotic_regret_probability}. Finally, if Equation~\eqref{eq:decision_rule} holds, then $\{\mathrm{Regret}>\Delta+\eta\}\subseteq\{\mathrm{Regret}>2R_{\max}\max_j\mu_j(\tau)+\eta\}$, so $\Pr(\mathrm{Regret}>\Delta+\eta)\to0$.
\end{proof}

\section{Detailed Per-ESS-Bin Results for Experiment 1}
\label{app:detailed_results}

\par{
This appendix breaks down Experiment 1 by the realized ESS, the ESS of the estimated mixture weights (Equation~\eqref{eq:ess}), whereas Table~\ref{tab:method_comparison_regimes} groups trials by ESS target: Top-1 accuracy (Table~\ref{tab:appendix_acc1}), Top-1 regret (Table~\ref{tab:appendix_regret}), and Value MAE\% (Table~\ref{tab:appendix_mae}). We additionally report a per-policy Value MAE\% breakdown by candidate value multiplier $v_k$ (Table~\ref{tab:per_policy_mae}).
We include the same six estimators from Section~\ref{sec:estimators}: IPS, Unweighted DM, Weighted DM, Per-policy Weighted DM, Unweighted DR, and Weighted DR.
Each row of Tables~\ref{tab:appendix_acc1}--\ref{tab:appendix_mae} reports the mean $\pm$ standard error (SE = std$/\sqrt{n}$) over trials assigned to the corresponding realized-ESS bin. The rounded bin labels from $0.1$ through $1.0$ are all populated across the $350$ trials.
Within each row, the best method mean is shown in \textbf{bold}.
}

\begin{table*}[t]
\centering
\caption{Top-1 accuracy $\uparrow$ by realized-ESS bin (mean $\pm$ SE). Best mean in each row is shown in \textbf{bold}.}
\label{tab:appendix_acc1}
\small
\setlength{\tabcolsep}{3pt}
\renewcommand{\arraystretch}{1.05}
\resizebox{\textwidth}{!}{%
\begin{tabular}{cc cccccc}
\toprule
ESS & $n$ & IPS & \shortstack{Unweighted\\DM} & \shortstack{Weighted DM\\(ours)} & \shortstack{Per-policy\\Weighted DM} & \shortstack{Unweighted\\DR} & \shortstack{Weighted DR\\(ours)} \\
\midrule
0.1 & 22 & $0.500\,{\pm}\,0.107$ & $0.409\,{\pm}\,0.105$ & $0.727\,{\pm}\,0.095$ & $\mathbf{0.818}\,{\pm}\,0.082$ & $0.727\,{\pm}\,0.095$ & $0.773\,{\pm}\,0.089$ \\
0.2 & 26 & $0.615\,{\pm}\,0.095$ & $0.654\,{\pm}\,0.093$ & $0.808\,{\pm}\,0.077$ & $0.692\,{\pm}\,0.091$ & $0.923\,{\pm}\,0.052$ & $\mathbf{0.962}\,{\pm}\,0.038$ \\
0.3 & 35 & $0.657\,{\pm}\,0.080$ & $0.600\,{\pm}\,0.083$ & $0.714\,{\pm}\,0.076$ & $0.714\,{\pm}\,0.076$ & $0.743\,{\pm}\,0.074$ & $\mathbf{0.857}\,{\pm}\,0.059$ \\
0.4 & 26 & $0.577\,{\pm}\,0.097$ & $0.769\,{\pm}\,0.083$ & $0.808\,{\pm}\,0.077$ & $0.846\,{\pm}\,0.071$ & $\mathbf{0.885}\,{\pm}\,0.063$ & $\mathbf{0.885}\,{\pm}\,0.063$ \\
0.5 & 30 & $0.567\,{\pm}\,0.090$ & $0.900\,{\pm}\,0.055$ & $0.933\,{\pm}\,0.046$ & $0.667\,{\pm}\,0.086$ & $0.933\,{\pm}\,0.046$ & $\mathbf{1.000}\,{\pm}\,0.000$ \\
0.6 & 34 & $0.618\,{\pm}\,0.083$ & $0.882\,{\pm}\,0.055$ & $0.882\,{\pm}\,0.055$ & $0.824\,{\pm}\,0.065$ & $0.882\,{\pm}\,0.055$ & $\mathbf{0.971}\,{\pm}\,0.029$ \\
0.7 & 50 & $0.680\,{\pm}\,0.066$ & $0.920\,{\pm}\,0.038$ & $\mathbf{0.980}\,{\pm}\,0.020$ & $0.860\,{\pm}\,0.049$ & $0.960\,{\pm}\,0.028$ & $0.960\,{\pm}\,0.028$ \\
0.8 & 41 & $0.659\,{\pm}\,0.074$ & $\mathbf{1.000}\,{\pm}\,0.000$ & $\mathbf{1.000}\,{\pm}\,0.000$ & $0.902\,{\pm}\,0.046$ & $\mathbf{1.000}\,{\pm}\,0.000$ & $\mathbf{1.000}\,{\pm}\,0.000$ \\
0.9 & 71 & $0.775\,{\pm}\,0.050$ & $0.944\,{\pm}\,0.027$ & $0.972\,{\pm}\,0.020$ & $0.901\,{\pm}\,0.035$ & $\mathbf{1.000}\,{\pm}\,0.000$ & $0.986\,{\pm}\,0.014$ \\
1.0 & 15 & $0.600\,{\pm}\,0.126$ & $\mathbf{1.000}\,{\pm}\,0.000$ & $\mathbf{1.000}\,{\pm}\,0.000$ & $0.933\,{\pm}\,0.064$ & $\mathbf{1.000}\,{\pm}\,0.000$ & $\mathbf{1.000}\,{\pm}\,0.000$ \\
\bottomrule
\end{tabular}
}%
\end{table*}

\begin{table*}[t]
\centering
\caption{Top-1 regret $\downarrow$ by realized-ESS bin (mean $\pm$ SE). Best mean in each row is shown in \textbf{bold}.}
\label{tab:appendix_regret}
\small
\setlength{\tabcolsep}{3pt}
\renewcommand{\arraystretch}{1.05}
\resizebox{\textwidth}{!}{%
\begin{tabular}{cc cccccc}
\toprule
ESS & $n$ & IPS & \shortstack{Unweighted\\DM} & \shortstack{Weighted DM\\(ours)} & \shortstack{Per-policy\\Weighted DM} & \shortstack{Unweighted\\DR} & \shortstack{Weighted DR\\(ours)} \\
\midrule
0.1 & 22 & $0.0041\,{\pm}\,0.0018$ & $0.0094\,{\pm}\,0.0026$ & $0.0034\,{\pm}\,0.0016$ & $0.0048\,{\pm}\,0.0024$ & $0.0042\,{\pm}\,0.0023$ & $\mathbf{0.0012}\,{\pm}\,0.0006$ \\
0.2 & 26 & $0.0054\,{\pm}\,0.0019$ & $0.0048\,{\pm}\,0.0017$ & $0.0015\,{\pm}\,0.0007$ & $0.0023\,{\pm}\,0.0012$ & $0.0010\,{\pm}\,0.0007$ & $\mathbf{0.0005}\,{\pm}\,0.0005$ \\
0.3 & 35 & $0.0042\,{\pm}\,0.0016$ & $0.0035\,{\pm}\,0.0015$ & $0.0013\,{\pm}\,0.0005$ & $0.0026\,{\pm}\,0.0012$ & $0.0012\,{\pm}\,0.0004$ & $\mathbf{0.0005}\,{\pm}\,0.0002$ \\
0.4 & 26 & $0.0032\,{\pm}\,0.0011$ & $0.0012\,{\pm}\,0.0006$ & $0.0004\,{\pm}\,0.0003$ & $0.0008\,{\pm}\,0.0005$ & $\mathbf{0.0002}\,{\pm}\,0.0001$ & $\mathbf{0.0002}\,{\pm}\,0.0001$ \\
0.5 & 30 & $0.0054\,{\pm}\,0.0016$ & $0.0004\,{\pm}\,0.0003$ & $0.0004\,{\pm}\,0.0003$ & $0.0028\,{\pm}\,0.0010$ & $0.0003\,{\pm}\,0.0002$ & $\mathbf{0.0000}\,{\pm}\,0.0000$ \\
0.6 & 34 & $0.0040\,{\pm}\,0.0013$ & $0.0004\,{\pm}\,0.0003$ & $0.0002\,{\pm}\,0.0001$ & $0.0018\,{\pm}\,0.0009$ & $0.0002\,{\pm}\,0.0001$ & $\mathbf{0.0000}\,{\pm}\,0.0000$ \\
0.7 & 50 & $0.0035\,{\pm}\,0.0013$ & $0.0003\,{\pm}\,0.0002$ & $\mathbf{0.0000}\,{\pm}\,0.0000$ & $0.0008\,{\pm}\,0.0004$ & $0.0001\,{\pm}\,0.0000$ & $0.0001\,{\pm}\,0.0000$ \\
0.8 & 41 & $0.0048\,{\pm}\,0.0016$ & $\mathbf{0.0000}\,{\pm}\,0.0000$ & $\mathbf{0.0000}\,{\pm}\,0.0000$ & $0.0006\,{\pm}\,0.0003$ & $\mathbf{0.0000}\,{\pm}\,0.0000$ & $\mathbf{0.0000}\,{\pm}\,0.0000$ \\
0.9 & 71 & $0.0025\,{\pm}\,0.0007$ & $0.0004\,{\pm}\,0.0003$ & $\mathbf{0.0000}\,{\pm}\,0.0000$ & $0.0005\,{\pm}\,0.0002$ & $\mathbf{0.0000}\,{\pm}\,0.0000$ & $\mathbf{0.0000}\,{\pm}\,0.0000$ \\
1.0 & 15 & $0.0045\,{\pm}\,0.0021$ & $\mathbf{0.0000}\,{\pm}\,0.0000$ & $\mathbf{0.0000}\,{\pm}\,0.0000$ & $0.0002\,{\pm}\,0.0002$ & $\mathbf{0.0000}\,{\pm}\,0.0000$ & $\mathbf{0.0000}\,{\pm}\,0.0000$ \\
\bottomrule
\end{tabular}
}%
\end{table*}

\begin{table*}[t]
\centering
\caption{Value MAE\% $\downarrow$ by realized-ESS bin (mean $\pm$ SE). Best mean in each row is shown in \textbf{bold}.}
\label{tab:appendix_mae}
\scriptsize
\setlength{\tabcolsep}{3pt}
\renewcommand{\arraystretch}{1.05}
\begin{tabular}{cc cccccc}
\toprule
ESS & $n$ & IPS & \shortstack{Unweighted\\DM} & \shortstack{Weighted DM\\(ours)} & \shortstack{Per-policy\\Weighted DM} & \shortstack{Unweighted\\DR} & \shortstack{Weighted DR\\(ours)} \\
\midrule
0.1 & 22 & $8.14\,{\pm}\,0.63$ & $12.22\,{\pm}\,1.32$ & $3.62\,{\pm}\,0.46$ & $3.91\,{\pm}\,0.50$ & $5.26\,{\pm}\,0.82$ & $\mathbf{3.45}\,{\pm}\,0.54$ \\
0.2 & 26 & $7.25\,{\pm}\,0.36$ & $8.47\,{\pm}\,0.89$ & $2.70\,{\pm}\,0.20$ & $3.40\,{\pm}\,0.42$ & $2.88\,{\pm}\,0.35$ & $\mathbf{1.75}\,{\pm}\,0.24$ \\
0.3 & 35 & $6.95\,{\pm}\,0.50$ & $5.93\,{\pm}\,0.51$ & $2.40\,{\pm}\,0.18$ & $3.42\,{\pm}\,0.34$ & $1.75\,{\pm}\,0.18$ & $\mathbf{1.19}\,{\pm}\,0.10$ \\
0.4 & 26 & $7.04\,{\pm}\,0.53$ & $3.87\,{\pm}\,0.39$ & $1.61\,{\pm}\,0.16$ & $2.56\,{\pm}\,0.24$ & $1.17\,{\pm}\,0.17$ & $\mathbf{0.59}\,{\pm}\,0.07$ \\
0.5 & 30 & $5.28\,{\pm}\,0.40$ & $2.82\,{\pm}\,0.28$ & $2.17\,{\pm}\,0.17$ & $3.24\,{\pm}\,0.30$ & $0.80\,{\pm}\,0.10$ & $\mathbf{0.62}\,{\pm}\,0.06$ \\
0.6 & 34 & $6.23\,{\pm}\,0.43$ & $2.45\,{\pm}\,0.24$ & $1.64\,{\pm}\,0.14$ & $2.75\,{\pm}\,0.23$ & $0.73\,{\pm}\,0.07$ & $\mathbf{0.40}\,{\pm}\,0.04$ \\
0.7 & 50 & $5.85\,{\pm}\,0.32$ & $1.88\,{\pm}\,0.15$ & $1.92\,{\pm}\,0.13$ & $2.74\,{\pm}\,0.19$ & $0.54\,{\pm}\,0.04$ & $\mathbf{0.36}\,{\pm}\,0.03$ \\
0.8 & 41 & $4.54\,{\pm}\,0.29$ & $1.44\,{\pm}\,0.12$ & $1.79\,{\pm}\,0.12$ & $2.42\,{\pm}\,0.15$ & $0.38\,{\pm}\,0.04$ & $\mathbf{0.24}\,{\pm}\,0.02$ \\
0.9 & 71 & $4.10\,{\pm}\,0.21$ & $1.67\,{\pm}\,0.09$ & $2.07\,{\pm}\,0.10$ & $2.69\,{\pm}\,0.10$ & $0.30\,{\pm}\,0.02$ & $\mathbf{0.22}\,{\pm}\,0.01$ \\
1.0 & 15 & $4.59\,{\pm}\,0.45$ & $1.98\,{\pm}\,0.26$ & $1.92\,{\pm}\,0.18$ & $2.30\,{\pm}\,0.22$ & $0.23\,{\pm}\,0.02$ & $\mathbf{0.19}\,{\pm}\,0.01$ \\
\bottomrule
\end{tabular}
\end{table*}

\begin{table*}[t]
\centering
\caption{Per-policy Value MAE\% $\downarrow$ by requested candidate value multiplier $v_k$ (mean $\pm$ SE). The column $n$ gives the number of trials containing the candidate. Best mean in each row is shown in \textbf{bold}.}
\label{tab:per_policy_mae}
\scriptsize
\setlength{\tabcolsep}{3pt}
\renewcommand{\arraystretch}{1.05}
\begin{tabular}{cc cccccc}
\toprule
$v_k$ & $n$ & IPS & \shortstack{Unweighted\\DM} & \shortstack{Weighted DM\\(ours)} & \shortstack{Per-policy\\Weighted DM} & \shortstack{Unweighted\\DR} & \shortstack{Weighted DR\\(ours)} \\
\midrule
0.95 & 350 & $5.11\,{\pm}\,0.23$ & $3.90\,{\pm}\,0.33$ & $2.13\,{\pm}\,0.09$ & $3.37\,{\pm}\,0.12$ & $1.18\,{\pm}\,0.16$ & $\mathbf{0.74}\,{\pm}\,0.09$ \\
1.00 & 350 & $5.44\,{\pm}\,0.25$ & $3.52\,{\pm}\,0.25$ & $2.04\,{\pm}\,0.08$ & $3.31\,{\pm}\,0.14$ & $1.02\,{\pm}\,0.09$ & $\mathbf{0.66}\,{\pm}\,0.06$ \\
1.05 & 350 & $5.32\,{\pm}\,0.23$ & $3.93\,{\pm}\,0.29$ & $2.02\,{\pm}\,0.09$ & $3.23\,{\pm}\,0.13$ & $1.21\,{\pm}\,0.13$ & $\mathbf{0.81}\,{\pm}\,0.09$ \\
1.10 & 343 & $6.42\,{\pm}\,0.27$ & $3.92\,{\pm}\,0.32$ & $2.24\,{\pm}\,0.09$ & $2.86\,{\pm}\,0.13$ & $1.36\,{\pm}\,0.18$ & $\mathbf{0.81}\,{\pm}\,0.10$ \\
1.15 & 336 & $6.13\,{\pm}\,0.25$ & $3.11\,{\pm}\,0.22$ & $2.22\,{\pm}\,0.08$ & $2.53\,{\pm}\,0.12$ & $0.98\,{\pm}\,0.09$ & $\mathbf{0.65}\,{\pm}\,0.06$ \\
1.20 & 294 & $5.62\,{\pm}\,0.28$ & $3.67\,{\pm}\,0.28$ & $2.32\,{\pm}\,0.10$ & $2.26\,{\pm}\,0.11$ & $1.20\,{\pm}\,0.13$ & $\mathbf{0.78}\,{\pm}\,0.08$ \\
\bottomrule
\end{tabular}
\end{table*}

\par{
Across realized-ESS bins, Weighted DR has the lowest Value MAE\% in every row, decreasing overall from $3.45$ at ESS $0.1$ to $0.19$ at ESS $1.0$. It also has the lowest regret in the three lowest-ESS bins, while multiple reward-model estimators reach near-zero regret as overlap increases. Top-1 accuracy improves overall with overlap, although the leading estimator varies by bin: Per-policy Weighted DM is strongest at ESS $0.1$, Weighted DR leads at ESS $0.2$, $0.3$, $0.5$, and $0.6$, and Weighted DM leads at ESS $0.7$. The per-policy breakdown shows that Weighted DR has the lowest Value MAE\% for every candidate value multiplier, while Weighted DM improves on Unweighted DM for all six candidates.
}

\section{Synthetic Experiment 2: Sensitivity to Importance Weight Clipping}
\label{app:exp2_clipping}

\par{
A practical question for any density-ratio weighting scheme is how aggressively to clip the estimated weights.
Clipping truncates the upper tail of the density-ratio weights, attenuating the correction for impressions far from the logging distribution; at $c=1$ all up-weighting is removed.
We study this tradeoff across different levels of policy-logging overlap, using the same simulation setup as Experiment~1.
}

\par{
Because policies are kernel scoring functions in the same feature space, the true mixture density ratio $\Omega(z) = p_{\mathrm{mix}}(z)/p_0(z)$ can be computed analytically from kernel scores, without a domain classifier.
We use the ESS of these \emph{kernel-based} weights as the ground-truth overlap measure, decoupling it from both the domain classifier and the clip threshold.
We run 50 seeds $\times$ 7 ESS levels, each evaluated at $c \in \{1, 5, 10, 20, 100\}$; the same domain classifier is trained once per trial and the five clip values are applied post hoc.
}

\par{
Figure~\ref{fig:clip_sensitivity} shows that Value MAE\% is largely insensitive to the clip threshold for $c \in \{10, 20, 100\}$: in every kernel-ESS bin these three settings differ by at most $0.52$ percentage points, while at kernel ESS $0.2$--$0.4$ they reduce Value MAE\% by $4.1$--$8.2$ points relative to the unweighted model.
From kernel ESS $0.5$ upward, $c=5$ matches or slightly improves on these settings, but at kernel ESS $0.2$--$0.4$ its Value MAE\% is higher (e.g., $3.20$ versus $2.10$--$2.62$ at kernel ESS $0.3$).
The exception is $c=1$: at kernel ESS $0.3$--$0.4$ its Value MAE\% ($6.24$--$6.55$) is close to that of the unweighted model ($6.69$--$7.27$), whereas at kernel ESS $0.9$--$1.0$ it attains the lowest Value MAE\% ($1.54$ and $1.38$, versus $1.90$--$2.36$ for the unweighted model and the other clip values).
In Top-1 regret and accuracy, $c=1$ is likewise closest to the unweighted model at kernel ESS $0.3$--$0.4$; at kernel ESS $0.2$, $c=5$ and $c=20$ have Top-1 regret $0.0137$ versus $0.0036$ for the other clip values.
The clip value used in Experiment~1, $c=100$, lies in this insensitive range.
}

\begin{figure}[ht]
\centering
\includegraphics[width=\textwidth]{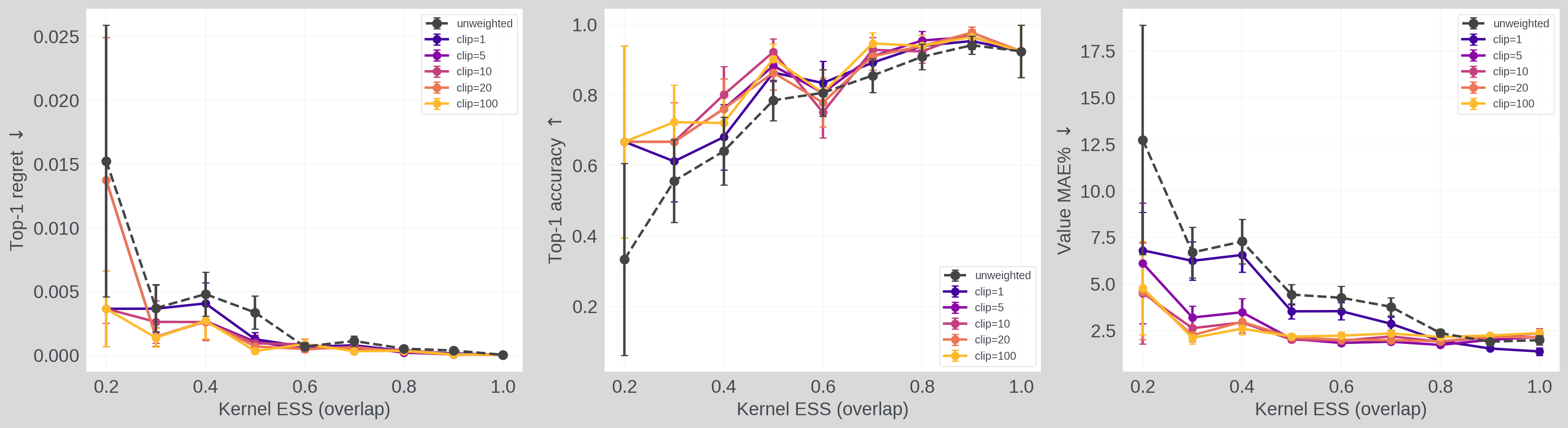}
\caption{Experiment 2: sensitivity of policy-selection metrics to importance-weight clipping as a function of kernel ESS (oracle, unclipped). Each point shows the mean $\pm$ standard error of the trials in one realized kernel-ESS bin (50 seeds $\times$ 7 ESS targets in total); the unweighted baseline is dashed black. Value MAE\% curves for $c \in \{10, 20, 100\}$ are tightly clustered; $c=1$ removes all up-weighting. $\uparrow$: higher is better; $\downarrow$: lower is better.}
\label{fig:clip_sensitivity}
\end{figure}

\section{Domain Classifier Reliability}
\label{sec:dc_reliability_experiment}

\par{
All estimators except Unweighted DM use density-ratio weights computed from the domain classifier's predicted probabilities (Section~\ref{sec:domain_classifier}), so their results depend on the quality of the classifier. We test this dependence by removing the classifier's calibration step and by masking the ad features it sees, and measure the effect on the downstream metrics of Experiment~1: Top-1 accuracy, Top-1 regret, and Value MAE\%. Each classifier variant below is evaluated on the same $350$ trials of Experiment~1, with everything else fixed:
\begin{itemize}
\setlength{\itemsep}{0pt}
\item \textbf{Experiment~1 DC}: the classifier of Experiment~1 (Appendix~\ref{app:per_trial_procedure}), trained on its own fold, disjoint from the reward model's (cross-fitting), and temperature-calibrated on 5-fold out-of-fold logits.
\item \textbf{Uncalibrated DC}: Experiment~1 DC without temperature calibration.
\item \textbf{Feature Masked DC}: Experiment~1 DC with one, two, or all three ad features set to zero ($33\%$, $67\%$, and $100\%$ ad masked). With all three masked, only the user features remain, and these are the same under every policy, so the classifier cannot tell the policies apart.
\end{itemize}
Table~\ref{tab:dc_estimator_performance} shows the results. Unweighted DM does not use the classifier, so its results are omitted here. Experiment~1 DC and Uncalibrated DC perform about the same for every estimator, so calibration has little effect on the downstream metrics. Masking the ad features degrades them: as $33\%$, $67\%$, and $100\%$ of the ad features are masked, the Top-1 accuracy of Weighted DM falls from $0.90$ to $0.86$, $0.83$, and $0.84$, that of Weighted DR from $0.94$ to $0.86$, $0.81$, and $0.84$, and that of IPS from $0.62$ to $0.27$. Per-policy Weighted DM has lower Top-1 accuracy than Weighted DM for every variant except full masking, where both reach $0.84$. With all ad features masked, the weighted estimators end close to their unweighted counterparts.
}

\begin{table*}[t]
\centering
\caption{Estimator performance under five domain-classifier variants over $50$ seeds. Each estimator cell reports Top-1 accuracy / Top-1 regret / Value MAE\% (Acc / Reg / MAE\%). Lower Reg and MAE\% are better; higher Acc is better. Best displayed value for each metric within each estimator column is in \textbf{bold}; ties are all bolded.}
\label{tab:dc_estimator_performance}
\small
\setlength{\tabcolsep}{3pt}
\renewcommand{\arraystretch}{1.05}
\resizebox{\textwidth}{!}{%
\begin{tabular}{l ccccc}
\toprule
DC variant & \shortstack{IPS \\ Acc $\uparrow$ / Reg $\downarrow$ / MAE\% $\downarrow$} & \shortstack{Weighted DM \\ Acc $\uparrow$ / Reg $\downarrow$ / MAE\% $\downarrow$} & \shortstack{Per-policy Weighted DM \\ Acc $\uparrow$ / Reg $\downarrow$ / MAE\% $\downarrow$} & \shortstack{Unweighted DR \\ Acc $\uparrow$ / Reg $\downarrow$ / MAE\% $\downarrow$} & \shortstack{Weighted DR \\ Acc $\uparrow$ / Reg $\downarrow$ / MAE\% $\downarrow$} \\
\midrule
Experiment~1 DC & \textbf{0.62} / 0.0052 / 5.56 & \textbf{0.90} / \textbf{0.0005} / \textbf{2.16} & 0.82 / 0.0016 / 2.92 & \textbf{0.90} / 0.0008 / 1.17 & \textbf{0.94} / \textbf{0.0002} / 0.73 \\
Uncalibrated DC & 0.59 / \textbf{0.0051} / \textbf{5.46} & \textbf{0.90} / 0.0006 / 2.47 & \textbf{0.84} / \textbf{0.0012} / 2.97 & 0.89 / \textbf{0.0007} / \textbf{1.14} & 0.93 / 0.0003 / \textbf{0.71} \\
Feature Masked DC ($33\%$) & 0.55 / 0.0106 / 7.85 & 0.86 / 0.0014 / 2.52 & 0.83 / 0.0014 / \textbf{2.61} & 0.86 / 0.0013 / 2.30 & 0.86 / 0.0012 / 2.41 \\
Feature Masked DC ($67\%$) & 0.33 / 0.0271 / 9.53 & 0.83 / 0.0019 / 3.09 & 0.80 / 0.0022 / 2.99 & 0.82 / 0.0017 / 3.48 & 0.81 / 0.0020 / 3.83 \\
Feature Masked DC ($100\%$) & 0.27 / 0.0315 / 10.45 & 0.84 / 0.0018 / 3.43 & \textbf{0.84} / 0.0015 / 3.50 & 0.83 / 0.0019 / 4.79 & 0.84 / 0.0018 / 4.56 \\
\bottomrule
\end{tabular}%
}
\end{table*}

\par{
Table~\ref{tab:dc_reliability_diagnostics} reports five measures of the classifier itself, to test whether they reveal this degradation without the downstream metrics. \textbf{NLL} is the negative log-likelihood of the classifier's predicted probability for the true policy on held-out impressions. \textbf{Cand.\ acc.}\ is the fraction of held-out candidate impressions for which the classifier ranks the candidate that produced them first among the candidates. \textbf{Ratio corr.}\ is the rank correlation between the estimated and the true density ratios, averaged over candidates. \textbf{Est.\ ESS} is $\mathrm{ESS}/n$ of the estimated mixture weights (Equation~\eqref{eq:ess}). \textbf{Cold-start err.}\ is the absolute difference between the estimated and the true cold-start mass at $\tau=0.1$. Ratio correlation and cold-start error need the true policies, so they are available only in simulation. Removing calibration leaves all five measures essentially unchanged, consistent with its small downstream effect. Masking the ad features moves them with the downstream degradation: NLL rises from $1.71$ to $1.91$, candidate accuracy falls from $0.32$ to $0.18$, ratio correlation falls from $0.47$ to $0.24$ at $67\%$ masking (it is undefined at $100\%$, where the estimated weights are constant), and cold-start error rises from $0.16$ to $0.39$. The estimated ESS instead rises from $0.59$ to $1.00$: the weights flatten because the classifier loses the ad features, not because overlap improves. A high ESS therefore indicates good overlap only when the classifier can tell the policies apart. NLL and candidate accuracy flag when it cannot without the oracle, but cannot tell whether the policies are genuinely similar or the features are missing; in simulation, ratio correlation and cold-start error separate these cases.
}

\begin{table*}[t]
\centering
\caption{Domain-classifier reliability diagnostics for the five classifier variants. Oracle ratio correlation and cold-start error are available only in simulation; the ratio correlation is undefined under full masking, where the estimated weights are constant. Lower NLL and cold-start error are better; higher candidate accuracy and ratio correlation are better. Best displayed value for each directional metric is in \textbf{bold}; ties are all bolded.}
\label{tab:dc_reliability_diagnostics}
\scriptsize
\setlength{\tabcolsep}{3pt}
\renewcommand{\arraystretch}{1.05}
\begin{tabular}{l ccc c c}
\toprule
DC variant & NLL $\downarrow$ & Cand. acc. $\uparrow$ & Ratio corr. $\uparrow$ & Est. ESS & Cold-start err. $\downarrow$ \\
\midrule
Experiment~1 DC & \textbf{1.71} & \textbf{0.32} & \textbf{0.47} & 0.59 & \textbf{0.16} \\
Uncalibrated DC & \textbf{1.71} & \textbf{0.32} & \textbf{0.47} & 0.58 & \textbf{0.16} \\
Feature Masked DC ($33\%$) & 1.77 & 0.29 & 0.37 & 0.69 & 0.21 \\
Feature Masked DC ($67\%$) & 1.84 & 0.24 & 0.24 & 0.81 & 0.27 \\
Feature Masked DC ($100\%$) & 1.91 & 0.18 & -- & 1.00 & 0.39 \\
\bottomrule
\end{tabular}
\end{table*}

\section{Production Data: Per-Policy Breakdown}
\label{app:prod_detailed}

\par{
We present the per-policy results of the production experiment in Section~\ref{sec:exp_prod}. Table~\ref{tab:prod_per_policy} reports the Value MAE\% of each estimator for each policy, where $P_1$ is the logging policy and $P_2$--$P_6$ are the five candidates; as in Table~\ref{tab:prod_results}, estimated weights are clipped at $c=20$ and all values are averaged over $50$ random subsamples. Among the DM estimators, Weighted DM has a lower Value MAE\% than Unweighted DM on every candidate ($10.91$--$12.24$ versus $13.30$--$15.19$) and a lower Value MAE\% than Per-policy Weighted DM on four of the five candidates, the exception being $P_6$ ($10.91$ versus $10.51$). IPS and both DR estimators have a lower Value MAE\% than all DM estimators on every candidate: Unweighted DR ranges from $1.90$ to $2.48$ and Weighted DR from $2.02$ to $2.81$, while IPS is lowest on $P_2$ and $P_4$ ($1.87$) but reaches $6.87$--$7.31$ on $P_3$, $P_5$, and $P_6$. On the logging policy $P_1$, the density ratio is identically one, so IPS and both DR estimators reduce to the mean of the logged rewards and have zero error by construction, and Per-policy Weighted DM coincides with Unweighted DM ($14.40$).
}

\begin{table}[t]
\centering
\caption{Production data: per-policy Value MAE\% for each estimator at clip $c=20$, mean $\pm$ SE over 50 random data subsamples (SE = std$/\sqrt{50}$). $P_1$ is the logging policy and $P_2$--$P_6$ are the candidate policies. Best per row among candidates in \textbf{bold}.}
\label{tab:prod_per_policy}
\footnotesize
\setlength{\tabcolsep}{4pt}
\renewcommand{\arraystretch}{1.05}
\begin{tabular}{l cccccc}
\toprule
Policy & IPS & Unweighted DM & Weighted DM & Per-policy Weighted DM & Unweighted DR & Weighted DR \\
\midrule
$P_1$ (logging) & $0.00\,{\pm}\,0.00$ & $14.40\,{\pm}\,1.45$ & $11.94\,{\pm}\,1.40$ & $14.40\,{\pm}\,1.45$ & $0.00\,{\pm}\,0.00$ & $0.00\,{\pm}\,0.00$ \\
\midrule
$P_2$ & $\mathbf{1.87}\,{\pm}\,0.18$ & $15.19\,{\pm}\,1.47$ & $12.16\,{\pm}\,1.44$ & $13.42\,{\pm}\,1.73$ & $1.90\,{\pm}\,0.20$ & $2.26\,{\pm}\,0.20$ \\
$P_3$ & $6.87\,{\pm}\,0.33$ & $13.37\,{\pm}\,1.33$ & $11.00\,{\pm}\,1.31$ & $12.16\,{\pm}\,1.27$ & $\mathbf{2.44}\,{\pm}\,0.25$ & $2.81\,{\pm}\,0.29$ \\
$P_4$ & $\mathbf{1.87}\,{\pm}\,0.22$ & $14.93\,{\pm}\,1.52$ & $12.24\,{\pm}\,1.39$ & $12.75\,{\pm}\,1.59$ & $2.09\,{\pm}\,0.20$ & $2.02\,{\pm}\,0.19$ \\
$P_5$ & $6.94\,{\pm}\,0.33$ & $13.30\,{\pm}\,1.32$ & $11.08\,{\pm}\,1.25$ & $11.36\,{\pm}\,1.25$ & $\mathbf{2.48}\,{\pm}\,0.24$ & $2.54\,{\pm}\,0.28$ \\
$P_6$ & $7.31\,{\pm}\,0.31$ & $13.38\,{\pm}\,1.27$ & $10.91\,{\pm}\,1.24$ & $10.51\,{\pm}\,1.30$ & $\mathbf{2.47}\,{\pm}\,0.29$ & $2.63\,{\pm}\,0.30$ \\
\bottomrule
\end{tabular}
\end{table}

\clearpage
\section*{NeurIPS Paper Checklist}

\begin{enumerate}

\item {\bf Claims}
    \item[] Question: Do the main claims made in the abstract and introduction accurately reflect the paper's contributions and scope?
    \item[] Answer: \answerYes{}
    \item[] Justification: The abstract and introduction state three contributions, all delivered in the paper: (1) a mixture-weighted reward model training method using density ratios estimated from a multiclass domain classifier, requiring no logged action propensities (Sections~\ref{sec:mixture_weighted} and~\ref{sec:domain_classifier}); (2) a high-probability bound on each candidate's replay Direct Method value error, separating an overlap term, a cold-start term, and a replay concentration term, and an asymptotic offline-vs-A/B routing rule derived from it (Theorem~\ref{thm:error_decomp} and Corollary~\ref{cor:decision_rule}, Appendix~\ref{app:error_analysis}); and (3) empirical evaluation on a configurable simulator (Section~\ref{sec:synthetic_experiments}) and on production traffic from a large-scale CTR system (Section~\ref{sec:exp_prod}), where the mixture-weighted reward model outperforms its unweighted counterpart in both settings. Scope and limitations of these claims (replay assumption, dependence on domain-classifier signal, asymptotic routing rule, no evaluation under live traffic) are explicitly stated in Section~\ref{sec:conclusions}.

\item {\bf Limitations}
    \item[] Question: Does the paper discuss the limitations of the work performed by the authors?
    \item[] Answer: \answerYes{}
    \item[] Justification: We provide a dedicated \textbf{Limitations} paragraph at the end of Section~\ref{sec:conclusions} (Conclusions). It states the replay assumption and the assumption that an impression's reward distribution does not depend on the policy that produced it, notes that computing the cold-start mass with the learned classifier requires calibration or error correction and that the weights help only when the representation retains policy-induced differences (Appendix~\ref{sec:dc_reliability_experiment}), and states that the routing rule is asymptotic and that its performance under live traffic is left to future work.

\item {\bf Theory assumptions and proofs}
    \item[] Question: For each theoretical result, does the paper provide the full set of assumptions and a complete (and correct) proof?
    \item[] Answer: \answerYes{}
    \item[] Justification: All theoretical results are numbered, cross-referenced, and stated together with their assumptions. The paper contains: Definition~\ref{def:overlap_coldstart} of the overlap and cold-start regions; Assumption~\ref{ass:consistency}, conditions \ref{ass:A1}--\ref{ass:A4} used by Theorem~\ref{thm:error_decomp}; Theorem~\ref{thm:error_decomp}, the high-probability replay Direct Method value-error bound, with the full proof in Appendix~\ref{app:error_analysis}; Corollary~\ref{cor:regret}, the top-1 selection regret bound, with proof in Appendix~\ref{app:regret_proof}; and Corollary~\ref{cor:decision_rule}, the asymptotic offline-vs-A/B routing rule, with proof in Appendix~\ref{app:decision_rule_proof}.

    \item {\bf Experimental result reproducibility}
    \item[] Question: Does the paper fully disclose all the information needed to reproduce the main experimental results of the paper to the extent that it affects the main claims and/or conclusions of the paper (regardless of whether the code and data are provided or not)?
    \item[] Answer: \answerYes{}
    \item[] Justification: For the synthetic experiments, the simulator design (user/ad feature distributions, oracle reward function, logging-policy and candidate-policy kernel parameterization), the per-trial pipeline, the value-multiplier and ESS-target grids, and the clip-threshold sweep are documented in detail in Appendix~\ref{app:simulator_details}, and the released code reproduces the synthetic results in Table~\ref{tab:method_comparison_regimes} and Appendices~\ref{app:detailed_results},~\ref{app:exp2_clipping} and \ref{sec:dc_reliability_experiment} (see the open-access answer below). For the production experiment, production data, embeddings, and trained production models will not be released; Section~\ref{sec:exp_prod} documents the evaluation protocol (one logging and five candidate policies, the empirical click-through rate as each policy's true value, $50$ random data subsamples, clip $c=20$), and the results are reported in main paper Table~\ref{tab:prod_results} with the per-policy breakdown in Appendix~\ref{app:prod_detailed}.

\item {\bf Open access to data and code}
    \item[] Question: Does the paper provide open access to the data and code, with sufficient instructions to faithfully reproduce the main experimental results, as described in supplemental material?
    \item[] Answer: \answerYes{}
    \item[] Justification: The simulator and synthetic evaluation code are publicly released at \url{https://github.com/facebookresearch/ad_ranking_simulation} (footnote in Section~\ref{sec:synthetic_experiments}). Its README specifies the environment and the exact commands that reproduce Table~\ref{tab:method_comparison_regimes} and the synthetic results in Appendices~\ref{app:detailed_results},~\ref{app:exp2_clipping} and \ref{sec:dc_reliability_experiment}. Production data, embeddings, and trained production models will not be released.

\item {\bf Experimental setting/details}
    \item[] Question: Does the paper specify all the training and test details (e.g., data splits, hyperparameters, how they were chosen, type of optimizer) necessary to understand the results?
    \item[] Answer: \answerYes{}
    \item[] Justification: For the synthetic experiments, Section~\ref{sec:synthetic_experiments} and the per-trial procedure in Appendix~\ref{app:per_trial_procedure} specify the data split ($70{,}000$ logged impressions in four disjoint folds: $20{,}000$ for domain-classifier training, $20{,}000$ for reward-model training, $10{,}000$ for calibration, and $20{,}000$ for evaluation), the architectures (one-hidden-layer MLPs with $64$ units for the domain classifier and the reward models), the optimizer (Adam with learning rate $0.01$ for the domain classifier and $10^{-3}$ for the reward models, batch size $512$), the training epochs ($4$ for the domain classifier, $30$ for the reward models), the temperature calibration ($5$-fold out-of-fold logits for the domain classifier, the calibration fold for the reward models), and the weight clip ($c=100$), with the clip sweep in Appendix~\ref{app:exp2_clipping}. For the production experiment, Section~\ref{sec:exp_prod} specifies the evaluation protocol ($50$ random data subsamples, clip $c=20$) and that the shared and per-policy reward models use the same data splits, architectures, optimization, calibration, and clipping settings.

\item {\bf Experiment statistical significance}
    \item[] Question: Does the paper report error bars suitably and correctly defined or other appropriate information about the statistical significance of the experiments?
    \item[] Answer: \answerYes{}
    \item[] Justification: Main paper Table~\ref{tab:method_comparison_regimes} reports means over the $350$ synthetic trials ($50$ seeds $\times$ $7$ ESS targets), grouped by ESS regime, and Table~\ref{tab:prod_results} reports means over $50$ random data subsamples. Standard errors of the mean (SE $=$ std$/\sqrt{n}$) are reported in the appendix breakdowns: Tables~\ref{tab:appendix_acc1}, \ref{tab:appendix_regret}, and~\ref{tab:appendix_mae} ($n$ is the number of trials in each realized-ESS bin), Table~\ref{tab:per_policy_mae} ($n$ is the number of trials containing the candidate), Figure~\ref{fig:clip_sensitivity}, and Table~\ref{tab:prod_per_policy} ($n=50$). The SE captures the per-trial randomness (oracle, users, ads, logging policy) for the synthetic experiments and the random data subsample for the production experiment.

\item {\bf Experiments compute resources}
    \item[] Question: For each experiment, does the paper provide sufficient information on the computer resources (type of compute workers, memory, time of execution) needed to reproduce the experiments?
    \item[] Answer: \answerYes{}
    \item[] Justification: The synthetic results were produced on CPU. On a laptop CPU, one seed takes about $3.5$ minutes for Table~\ref{tab:method_comparison_regimes} and $3$ minutes for Appendix~\ref{app:exp2_clipping} in a single process; a seed of Appendix~\ref{sec:dc_reliability_experiment} takes longer because it also trains the per-policy reward models in each of its five classifier variants. The $50$ seeds of each experiment are independent, so they can be split across processes. The production experiment runs on internal infrastructure, and its compute is not reported. The contribution is methodological and is not bottlenecked by specialized infrastructure.

\item {\bf Code of ethics}
    \item[] Question: Does the research conducted in the paper conform, in every respect, with the NeurIPS Code of Ethics \url{https://neurips.cc/public/EthicsGuidelines}?
    \item[] Answer: \answerYes{}
    \item[] Justification: The research is a methodological contribution to offline policy evaluation. The synthetic experiments use a simulator; the production experiment uses aggregated impression logs from a large-scale industrial CTR ranking system, with all per-policy results in the paper represented as $P_1$--$P_6$ and reported only as aggregate metrics. We do not release datasets, and do not introduce models intended for or applicable to surveillance, discrimination, or other prohibited use cases.

\item {\bf Broader impacts}
    \item[] Question: Does the paper discuss both potential positive societal impacts and negative societal impacts of the work performed?
    \item[] Answer: \answerNA{}
    \item[] Justification: The paper is foundational methodology that improves the reliability of one specific intermediate step (offline policy evaluation) in industrial ad-ranking pipelines. It does not introduce new ML capabilities, new data sources, or new deployment surfaces, and it is not tied to any particular downstream application beyond standard ad ranking. We do not foresee direct positive or negative societal impacts of the methodological contribution itself beyond the operational reliability gains within the host ad system.

\item {\bf Safeguards}
    \item[] Question: Does the paper describe safeguards that have been put in place for responsible release of data or models that have a high risk for misuse (e.g., pre-trained language models, image generators, or scraped datasets)?
    \item[] Answer: \answerNA{}
    \item[] Justification: The paper releases no data and no models, and the methodological contribution (offline policy evaluation via mixture-weighted reward modeling) does not produce assets with a high risk of misuse such as pre-trained generative models, image generators, or scraped datasets.

\item {\bf Licenses for existing assets}
    \item[] Question: Are the creators or original owners of assets (e.g., code, data, models), used in the paper, properly credited and are the license and terms of use explicitly mentioned and properly respected?
    \item[] Answer: \answerNA{}
    \item[] Justification: The paper does not use externally-licensed datasets or models. The synthetic experiments use a custom-built simulator developed by the authors, and the production experiment uses internal proprietary impression logs from the host institution's CTR ranking system; no third-party assets requiring license attribution are used.

\item {\bf New assets}
    \item[] Question: Are new assets introduced in the paper well documented and is the documentation provided alongside the assets?
    \item[] Answer: \answerYes{}
    \item[] Justification: The paper releases the simulator and synthetic evaluation code at \url{https://github.com/facebookresearch/ad_ranking_simulation}. Its README documents the environment, the commands that reproduce each synthetic table and figure, and the code layout, and the code is released under the CC BY-NC 4.0 license included in the repository.

\item {\bf Crowdsourcing and research with human subjects}
    \item[] Question: For crowdsourcing experiments and research with human subjects, does the paper include the full text of instructions given to participants and screenshots, if applicable, as well as details about compensation (if any)?
    \item[] Answer: \answerNA{}
    \item[] Justification: The paper does not involve crowdsourcing or research with human subjects. The synthetic experiments use a simulator with no real users; the production experiment analyzes pre-existing aggregated impression logs from a deployed industrial CTR ranking system without any new human-subject data collection.

\item {\bf Institutional review board (IRB) approvals or equivalent for research with human subjects}
    \item[] Question: Does the paper describe potential risks incurred by study participants, whether such risks were disclosed to the subjects, and whether Institutional Review Board (IRB) approvals (or an equivalent approval/review based on the requirements of your country or institution) were obtained?
    \item[] Answer: \answerNA{}
    \item[] Justification: The paper does not involve research with human subjects; no IRB or equivalent approval was required.

\item {\bf Declaration of LLM usage}
    \item[] Question: Does the paper describe the usage of LLMs if it is an important, original, or non-standard component of the core methods in this research? Note that if the LLM is used only for writing, editing, or formatting purposes and does \emph{not} impact the core methodology, scientific rigor, or originality of the research, declaration is not required.
    \item[] Answer: \answerNA{}
    \item[] Justification: The core method development in this research does not involve large language models as any important, original, or non-standard component. The proposed estimators (mixture-weighted reward modeling, density-ratio estimation via a multiclass domain classifier, doubly-robust estimation) are statistical and do not use LLMs.

\end{enumerate}

\end{document}